\documentclass[sigconf]{acmart}
\usepackage{multicol}
\usepackage{balance}
\usepackage{graphicx}
\usepackage{algorithm}
\usepackage{amsthm}

\usepackage[english]{babel}
\usepackage{amsthm}
\usepackage{url}
\usepackage{ulem}
\usepackage{color}
\usepackage{algorithmicx}
\usepackage{subfigure}
\usepackage{multirow} 
\usepackage[noend]{algpseudocode}

\AtBeginDocument{%
  }

\setcopyright{acmlicensed}
\copyrightyear{2024}
\acmYear{2024}
\acmDOI{XXXXXXX.XXXXXXX}

\acmISBN{978-1-4503-XXXX-X/18/06}

\begin{document}


\title{{Updateable Data-Driven Cardinality Estimator with Bounded Q-error}}



\author{Yingze Li}
\affiliation{%
  \institution{Harbin Institute of Technology}
  \country{China}
}
\email{23B903046@stu.hit.edu.cn}

\author{Xianglong Liu}
\affiliation{%
  \institution{Harbin Institute of Technology}
  \country{China}
}
\email{23S003029@stu.hit.edu.cn}

\author{Hongzhi Wang\footnotemark[1]}
\affiliation{%
  \institution{Harbin Institute of Technology}
  \country{China}
}
\email{wangzh@hit.edu.cn}

\author{Kaixin Zhang}
\affiliation{%
  \institution{Harbin Institute of Technology}
  \country{China}
}
\email{21B903037@stu.hit.edu.cn}

\author{Zixuan Wang }
\affiliation{%
  \institution{Harbin Institute of Technology}
  \country{China}
}
\email{2023113027@stu.hit.edu.cn}


\begin{abstract}
Modern cardinality estimators struggle with data updates. This research tackles this challenge within a single table.  We introduce ICE, an \uline{I}ndex-based \uline{C}ardinality \uline{E}stimator, the first data-driven estimator that enables instant, tuple-leveled updates.\looseness=-1

ICE has learned two key lessons from the multidimensional index and applied them to solve CE  in dynamic scenarios: (1) Index possesses the capability for swift training and seamless updating amidst vast multidimensional data.  (2) Index offers precise data distribution, staying synchronized with the latest database version. These insights endow the index with the ability to be a highly accurate, data-driven model that rapidly adapts to data updates and is resilient to out-of-distribution challenges during query testing. To enable a solitary index to support CE, we have crafted specific algorithms for training, updating, and estimating. Furthermore, we have analyzed the unbiasedness and variance of the estimation results.\looseness=-1

Extensive experiments demonstrate the superiority of ICE. ICE offers precise estimations and fast updates/constructions across diverse workloads.  Compared to state-of-the-art real-time query-driven models, ICE boasts superior accuracy (2-3 orders of magnitude more precise), faster updates ($4.7-6.9\times$ faster), and significantly reduced training time (up to 1-3 orders of magnitude faster).\looseness=-1

\end{abstract}

\received{20 February 2007}
\received[revised]{12 March 2009}
\received[accepted]{5 June 2009}

\maketitle

\footnotetext[1]{Corresponding author.}

\section{Introduction}

{C}ardinality {e}stimation (CE) is crucial for query optimization~\cite{thuBaseline}, predicting query selectivity without execution. Despite its importance, unresolved issues persist, causing estimation errors up to $10^4$ in both open-source and commercial systems~\cite{AreWeReady4CE}. These errors can deteriorate query plans and database performance~\cite{DBMSCE}. Thus, addressing CE is vital for optimizing query execution and enhancing database performance.

To tackle CE, data-driven~\cite{yang2019deep,hilprecht2019deepdb,CardIndex} and query-driven methods~\cite{MSCN_Kipf2018LearnedCE,ALECE} have emerged. Data-driven methods learn joint data distributions for accurate estimations, while query-driven models learn query to cardinality mappings. Data-driven models, without prior workload knowledge, excel in generalizing to unknown queries, promising broader applications.

However, both data-driven and query-driven paradigms face a common Achilles' heel, \textbf{data updates}. For query-driven models, the cardinalities of all known training queries may change after data updates, necessitating the re-execution of queries in the training set to obtain the true cardinalities and retraining the model on the new workload. For example, as shown in Table~\ref{Tab.T01}, MSCN~\cite{MSCN_Kipf2018LearnedCE} has to re-execute the query and be retrained after each data update, which results in a huge overhead  $2.5 \times 10^9 \mu s$ per update. For data-driven models, whenever the original data is updated, the joint distribution of the relational table changes accordingly~\cite{ALECE}. Compared with the traditional estimators like histograms~\cite{ViswanathPoosala1996ImprovedHF}, the learned data-driven models require a much slower finetuning process to adapt to the new database state~\cite{ALECE,AreWeReady4CE}. Although this process is prolonged, updating these models in a dynamically changing environment is indispensable. Using the stale model for prediction or estimation may result in errors as high as $10^3$ or even more, as evidenced by the stale Naru and MSCN's max Q-error in Table~\ref {Tab.T01}. This level of error cannot be overlooked. These challenges make it difficult for these models to adapt to real-world scenarios with real-time data updates.\looseness=-1


\begin{table*}[]
\caption{Modern cardinality estimators in dynamic environments. For models that update slowly, we use $\mathcal{N}$ to denote the \uline{N}ew fine-tuned model and $\mathcal{S}$ to denote the \uline{S}tale model. We report the result on the Insert-Heavy workload of the DMV~\cite{DMV} dataset. More details can be found in Section~\ref{Sec.EstEva}.}
\label{Tab.T01}
\small
\begin{tabular}{l||l|l|l|l|l}
\textbf{}& ICE(Ours) & MSCN & Naru & CardIndex & ALECE  \\ \hline
Model& Index Only & Conv Network & AR Network & AR Network + Index & Transformer + Hist  \\
Type& Data-driven & Query-driven & Data-driven & Data-driven & Query-driven  \\
Training Time & $ \textbf{13.8}s $ & $76(Train)+2514(Label)s$ & $ 2972s $ & $ 112(Deep) + 15.9(Index)s$ & $69(Train)+2514(Label)s$ \\
Update Time &  $ \textbf{6.2} \mu s/ tuple$ & $2.5 \times 10^9 \mu s/update $  & $453 \mu s/ tuple$ & $36.7(Deep) + 8.1(Index) \mu s/tuple$ & $19 \mu s/ tuple $     \\

Inference Time & $8.12 ms$ &  {$\textbf{0.71} ms$} & $17.6 ms$ & $8.46 ms$ & $1.51 ms$  \\
QMAX &  {$\textbf{12}$} &$  1\times10^3(\mathcal{N}) | 4\times10^3(\mathcal{S})$ & $23(\mathcal{N}) | 6\times10^3(\mathcal{S})$ & $3\times 10^3$ & $525$ 

\end{tabular}
\end{table*}



In order to deal with the CE  in dynamic scenarios, attempts have been made using the query-driven paradigm. ALECE~\cite{ALECE}, a query-driven model, utilizes transformers to learn the correlation between histogram and query representations. When data is updated, ALECE modifies the histogram representation without fine-tuning the transformer model. Although ALECE utilizes coarse data features such as histogram representation, when there is a certain distance between the query distribution in the testing set and that in the training set, ALECE still makes a significant prediction error over a magnitude of $525$, as shown in Table~\ref{Tab.T01}. This means its generalization ability on unseen queries still cannot reach the same level as data-driven models. 


The fundamental reason why existing deep data-driven models cannot effectively support data updates in databases lies in \textbf{the inefficiency of neural networks} in storing or learning the representations of database tuples.   For a trained neural network, the representation of a single data tuple is dispersed across every learned network parameter in each network layer.   Suppose that we want the network to "insert" new tuples or "forget" old ones. In this case, we have to iteratively perform gradient descent on each parameter in the network and update the parameters one by one~\cite{unlearn}.   This issue results in inefficient updates, and even with the parallel computing power of hundreds or thousands of cores on modern GPU hardware. As shown in Table~\ref{Tab.T01}, Naru's update latency is still only 453 microseconds per updated tuple. In contrast, the index stores different tuples in the leaf slots of corresponding subtrees, maintaining mutual isolation. For each insertion/deletion/modification operation, only local parameters of the scale of $log(N)$ will be modified in the subtrees of the model, where $N$ is the total size of the data. In contrast, the parameters in the rest of the model remain unchanged. Therefore, the index is highly efficient and widely deployed in updates for many dynamic database scenarios~\cite{jensen2004query,achakeev2013efficient,wu2010efficient}.



\textcolor{black}{We observe that index structure serve as an efficient model for both preserving lossless data distribution and supporting rapid data updates. These characteristics have sparked our reflection}: \textit{\textbf{Is it possible to leverage existing mature index structures in databases to create a fully index-driven cardinality estimator?}} Regarding estimation accuracy, as the index fully preserves the lossless distribution of data, we believe that such a cardinality estimator will be able to achieve high-precision predictions of query selectivity. Furthermore, in terms of update speed, given the rapid development of index technology over the past four decades~\cite{1979B+tree,graefe2011modern,kraska2018case,lmsfc}, we can expect this cardinality estimator to enable fast and efficient updates of data distributions, thus significantly enhancing the performance of database queries.

\textcolor{black}{However, relying solely on index structure in data-driven CE poses a challenge. It is difficult to imagine how an index can trim the sampling space online like an \uline{A}uto\uline{R}egressive(AR) model and achieves high estimation efficiency. Existing index-supported cardinality estimators like CardIndex~\cite{CardIndex} try a compromise by "gluing" AR network with learned index, but this falls short in accuracy and efficiency. CardIndex's accuracy suffers as it relies on neural networks for high-cardinality queries, but limited model parameters yield inaccurate probability density predictions, causing significant errors (see Table~\ref{Tab.T01}). In terms of efficiency, every database update requires a slow, thorough update of the neural network, which significantly slows down CardIndex's training and updating process, with neural network training and fine-tuning accounting for most of the time (see Table~\ref{Tab.T01}).}\looseness=-1


\textcolor{black}{To address these challenges, we draw inspiration from existing multidimensional indexing works~\cite{tropf1981multidimensional,lmsfc} by leveraging  data-skipping technique of multidimensional indexes to enhance sampling efficiency. We first utilize the data-skipping technique to filter the query space.} Then, we use the index to convert the filtered subregions into a compact rank space and sample points on the rank space.
At last, we use the index again to map those sampled values back to the original tuple representations and aggregate the results. In essence, we have devised an utterly index-driven CE methodology that attains excellent estimation accuracy and seamlessly facilitates instant updates encompassing insertions, deletions, and modifications. All the update operations can be accomplished within $O(log(N))$ time.\looseness=-1

The contributions of the paper are summarized below:

\textbf{S1.} We proposed ICE, an \uline{I}ndex-based \uline{C}ardinality \uline{E}stimator(ICE)  (Section~\ref{Sec.DetailedLayout}). It is the first high-precision, data-driven learned structure that supports instant data insertion/deletion/modification.

\textbf{S2.} We designed efficient bulk-loading  (Section~\ref{Sec.bulk-loading}), updating  (Section~\ref{Sec.CURD}) algorithms for ICE, enabling ICE to train rapidly from massive data and update quickly in dynamic scenes.\looseness=-1

\textbf{S3.} We designed an efficient CE algorithm based on ICE   (Section~\ref{Sec.CEALG}). The core idea is to sample in the filtered latent space, i.e., the rank space. We also analyzed the unbiasedness and variance of the method as well as the probability of the algorithm's estimation exceeding the preset maximum Q-error when predicting the cardinality of low-cardinality queries  (Section~\ref{Sec.Analysis}). By executing these low-cardinality queries with a fast index scan, we can bound the Q-error to any user-specified requirements. \looseness=-1

\textbf{S4.} Extensive experiments have demonstrated the superiority of  ICE  (Section~\ref{sec.EXP}). ICE has achieved rapid updates and accurate estimation in multiple datasets and various dynamic scenes. Compared to SOTA real-time query-driven models, ICE boasts 2-3 orders of magnitude higher accuracy, $4.7-6.9\times$ faster updates, and training time expedited by up to 1-3 orders of magnitude.

Due to space limitations, the scope of discussion in this work will be limited to the multidimensional CE  task within a single table. Adapting the multidimensional cardinality estimator from a single table to multiple tables is straightforward. We can perform a similar transformation like CardIndex~\cite{CardIndex}, which samples from the full outer join of all tables and then builds a multidimensional index on the materialized full outer join sample.



\section{Related Work}


\textbf{Modern cardinality estimators}:  Modern cardinality estimators can be broadly categorized into two paradigms: query-driven~\cite{LWNN,MSCN_Kipf2018LearnedCE,ALECE} and data-driven\cite{yang2019deep,hilprecht2019deepdb,CardIndex}. Query-driven estimators employ neural networks like MLP~\cite{LWNN}, CNN~\cite{MSCN_Kipf2018LearnedCE}, and Transformer~\cite{ALECE} to build regression models that predict query cardinalities by learning mappings from query representations to true cardinalities. Two challenges persist in this paradigm: data updates (necessitating re-acquiring training labels and retraining models after update~\cite{AreWeReady4CE}) and workload shifts (unseen workload patterns in test queries comparing to training data~\cite{yang2019deep}). To tackle data updates, ALECE~\cite{ALECE} introduces additional histogram features and leverages attention mechanisms to learn correlations between query representations and histogram encodings. Since histogram features can be swiftly updated, ALECE circumvents the need to re-acquire true cardinalities and retrain based on them, enhancing update speeds for query-driven models. On the other hand, data-driven estimators learn joint data distributions, utilizing statistical inference on \uline{S}um-\uline{P}roduct \uline{N}etwork (SPN)~\cite{hilprecht2019deepdb} or sampling on deep model~\cite{yang2019deep} to estimate cardinalities. The former often lags in estimation accuracy due to assumptions of independence among data attributes. Data-driven estimators are inherently robust against workload shifts~\cite{yang2019deep,AreWeReady4CE}. However, efficiently updating data-driven models remains challenging: SPN requires reconstruction when new correlations appears~\cite{hilprecht2019deepdb}, and deep models suffer from inefficient tuple representation storage in neural networks, necessitating entire network parameter updates via gradient descent~\cite{yang2019deep}. CardIndex~\cite{CardIndex} mitigates update issues by stacking a lightweight AR network with a learned index, delegating difficult small-cardinality queries to the index to search, and reserving neural networks for estimating large-cardinalities. This reduces the network size and enhances training and update speed. Nevertheless, CardIndex's reliance on AR networks at its root node means that it cannot fundamentally resolve data update challenges.

\textbf{Z-ordered multidimensional index}: To index multidimensional data, Z-ordering curves are commonly employed, mapping data to one dimension for indexing~\cite{lmsfc,ramsak2000integrating,jensen2004query,tropf1981multidimensional}. Handling range queries based on Z-order involves calculating $Z_{min}$ and $Z_{max}$ for the query box and scanning data within. However, this includes non-relevant data, causing gaps. Tropf et al~\cite{tropf1981multidimensional} proposes \textit{getBIGMIN} and \textit{getLITMAX} methods to skip these gaps efficiently in $O({n})$ time, where $n$ is the bit length. LMSFC~\cite{lmsfc} further optimizes Z-order indexing by switching bit order and leveraging optimal-1 split for range query performance improvements based on historical workloads.\looseness=-1

\section{Preliminary}

In this section, we will introduce some basic concepts of CE  and multidimensional indexing  (Section~\ref{sec.baseicConcepts}) and revisit the index structure from the perspective of CE (Section~\ref{sec.revist}).

\subsection{Basic Concepts}
\label{sec.baseicConcepts}
For a relational table $T$ containing $N$ tuples and $m$ attributes $\{A_1, A_2, \\ \dots, A_m\}$, where each attribute is encoded using $\beta$ binary bits, i.e., $A_i = \mathcal{B}_{i1}\mathcal{B}_{i2}\dots \mathcal{B}_{i\beta}$. \textcolor{black}{A query predicate $\theta$ can be viewed as a function that takes a tuple $t \in T$ as input and returns $\theta(t)=1$ if $t$ satisfies the predicate's conditions, and $\theta(t)=0$ otherwise.} The set of tuples in table $T$ that satisfy the predicate $\theta$ comprises the result set \textcolor{black}{$R = \{ t \in T : \theta(t) = 1 \}$}. \textcolor{black}{To expedite the retrieval of the result set $R$, an index relies on an auxiliary structure. Furthermore, CE necessitates the computation of the size of the result set, which is denoted by \textcolor{black}{$card(\theta) = |R|$}.}

The \uline{P}robability \uline{D}ensity \uline{F}unction(PDF) of a tuple $t$  reflects its proportion in the relational table $PDF(t) = o(t)/N$  where $o(t)$ represents the frequency of $t$ in table $T$. It is inherently linked to the query selectivity $sel(\theta)$, Consequently,

\begin{equation}
sel(\theta) = \sum_{t \in T} \theta(t) \times PDF(t)
\end{equation}

Given a table $T$ sorted under a specific bitwise ordering $\Omega$, expressed as $\Omega = \langle \mathcal{B}_{x_1}, \mathcal{B}_{x_2}, \ldots, \mathcal{B}_{x_{m\times \beta}}  \rangle$, the relationship $t_2 < t_1$ holds if and only if there exists a value of $\kappa \in [1,m\times\beta]$ such that $t_2.\mathcal{B}_{x_\kappa} < t_1.\mathcal{B}_{x_\kappa}$ while for all $j < \kappa$, we have $t_2.\mathcal{B}_{x_j} = t_1.\mathcal{B}_{x_j}$.

 Under such particular ordering $\Omega$, and the rank of tuple $t$, $r_t$ needs to count all the tuples whose keys are smaller than $t$.

\begin{equation}
r_t= \sum_{t_i <  t} o(t_i) =  \sum_{t_i <  t}  PDF(t_i)\times N
\end{equation}

The \uline{C}umulative \uline{D}istribution \uline{F}unction(CDF) is the result of scaling the rank of $t$ by $N$ times, namely $CDF(t) = r_t/N$.

To process multidimensional data more efficiently, we adopt the Z-order bit
ordering convention $\Omega_Z$, which is denoted as: $\Omega_Z =  \langle \mathcal{B}_{11}, \mathcal{B}_{21}, \ldots, \mathcal{B}_{{m1}}, \mathcal{B}_{12}, \mathcal{B}_{22}, \ldots, \mathcal{B}_{{m2}}, \ldots \mathcal{B}_{1\beta}, \mathcal{B}_{2\beta}, \ldots, \mathcal{B}_{{m\beta}}  \rangle$. To simplify the notation, we use $n$ to replace $m\beta$ in representing the total coding length, set the ordering notation as $\Omega_Z=\left \langle Z_1, Z_2, \dots Z_n \right \rangle$,  and denote 
 tuple $t$ under Z-ordering as $t=(z_1, z_2, \dots z_n) $.\looseness=-1

\subsection{Revisiting Index Structures from CE Perspectives }
\label{sec.revist}
An index is a precise learned model that takes the query key as the input and efficiently predicts its CDF~\cite{kraska2018case}. Theoretically, the PDF of a tuple is the derivative of its CDF (rank) as follows.

\begin{equation}
{PDF}(t) = \frac{d}{dt} {CDF}(t)   = \frac{1}{N} \times \frac{d}{dt} r_t
\label{eq:pdf_cdf_derivative}
\end{equation}

In practice, we can avoid the differentiation operation to obtain the PDF more ingeniously by utilizing an additional counter in leaf nodes of the index to maintain the frequency of tuples $o(t)$. And we use $PDF(t) = o(t)/N$ to obtain the PDF.

Therefore, not only limited to CardIndex~\cite{CardIndex}, any index can implement CDF-call and PDF-call within one point query.

\begin{figure}[htbp]
	\centering
	\includegraphics[width=8.5cm]{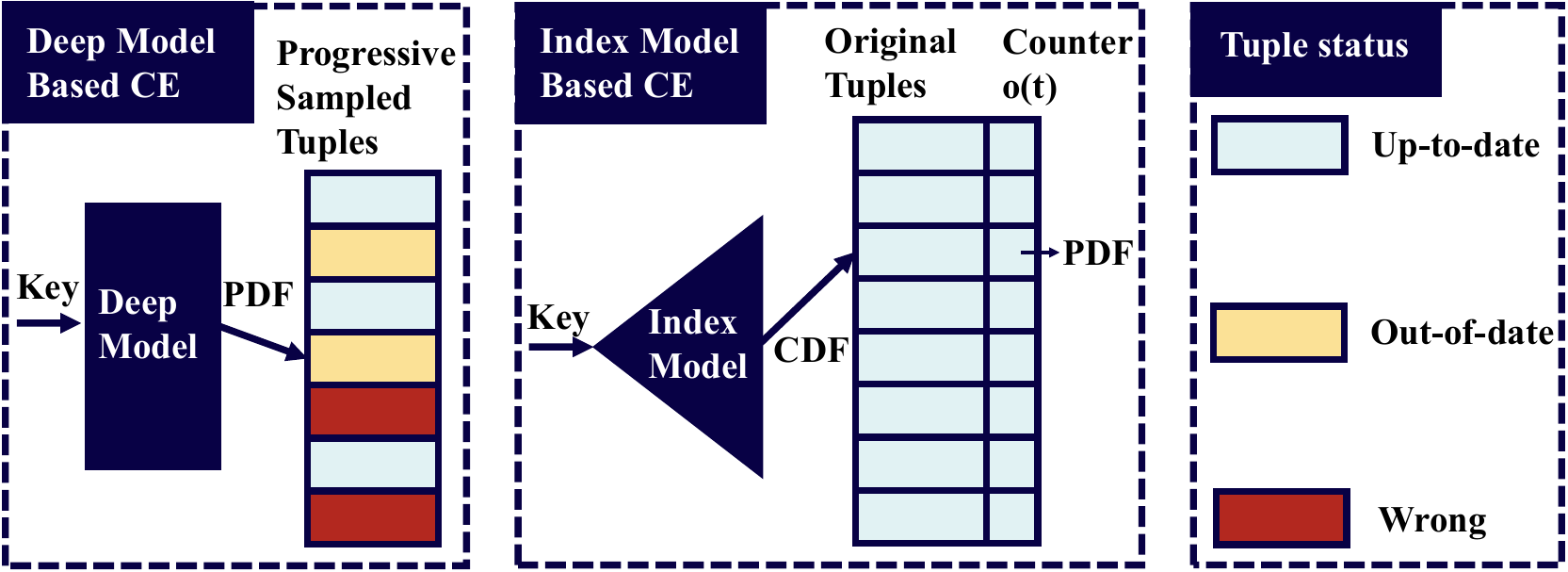}
	\caption{Deep network vs. index in obtaining PDF}
	\label{P1}
\end{figure}

From the perspective of CE, compared with existing deep models, the PDF obtained by index has at least two advantages over that obtained by deep learning models:\looseness=-1

\textbf{S1.Accurate}: The PDF information obtained through the index is lossless and will not cause prediction errors due to insufficient model training or too many distinct values in columns.

\textbf{S2.Fresh}: For updates in the database, we can directly perform physical insert/delete/modification on the index, ensuring that the model is \textbf{always up-to-date} when outputting PDF for estimation. However, deep models require an additional and slow fine-tuning process to synchronize the model with the updated database.\looseness=-1

\section{Index for Cardinality Estimation (ICE)}

\textcolor{black}{We observe that for updates, B+-tree and similar tree-shaped indexes can efficiently fine-tune their parameters within $log(N)$ complexity, making them well-suited for handling massive dynamic multidimensional data~\cite{graefe2011modern,LIPP}. In this section, we elaborate on adapting these prevalent structures for CE. The ICE structure is shown in Figure~\ref{P2}. It is similar to a B+-tree but maintains additional lightweight counters to assist CE. In this section, we introduce the layout of ICE in Section~\ref{Sec.DetailedLayout} and describe the bulk-loading algorithm in Section~\ref{Sec.bulk-loading}. We also elaborate on the support of ICE for real-time insertion, deletion, and modification in Section~\ref{Sec.CURD}. Finally, we leverage ICE in Section~\ref{Sec.Tup2CDF} to achieve the bijection from the key space to the rank space within $log(N)$ time, facilitating the subsequent CE task.}\looseness=-1

\begin{figure}[htbp]
	\centering
	\includegraphics[height=3.5cm]{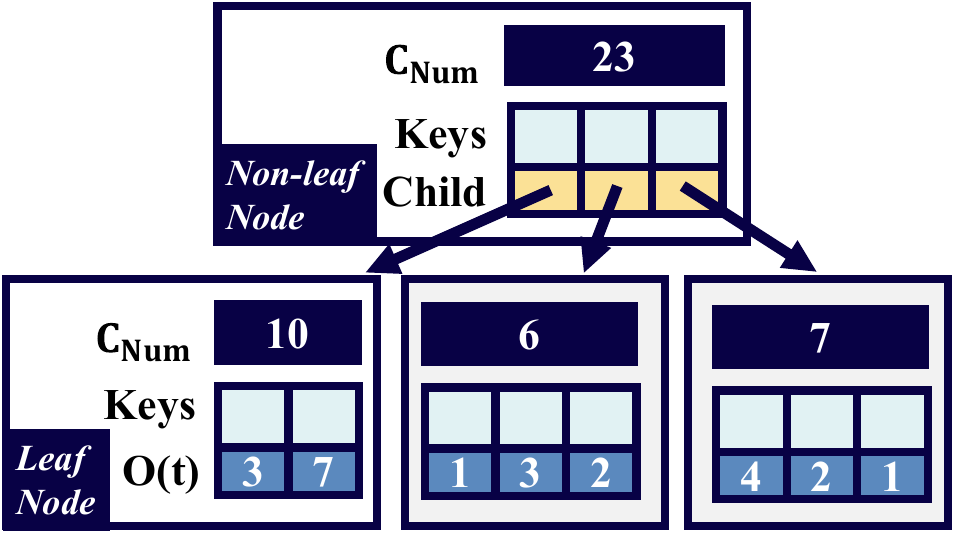}
	\caption{Layout of ICE}
	\label{P2}
        \vspace{-1em}
\end{figure}

\subsection{Index Layout}
\label{Sec.DetailedLayout}

\textcolor{black}{To make the estimator achieve $log(N)$ time update, we choose a tree-based index layout similar to the B+-tree~\cite{graefe2011modern,LIPP}. We also maintain two lightweight counters, $C_{Num}$ and $o(t)$ to guarantee that we can implement a fast bidirectional mapping from the Z-order representation of a tuple (key space) to the rank of a tuple (rank space) in Section~\ref{Sec.Tup2CDF}.  We define these two counters as follows.}\looseness=-1

\textbf{Tuple frequency counter $o(t)$.} For a given tuple $t$ stored in the leaf node, we maintain a counter $o(t)$ to get the frequency of tuple $t$. From $o(t)$, we can easily derive the PDF of tuple $t$  by the formula $PDF(t) = o(t)/N$ where $N$ is the total data size. This counter avoids additional derivative operations or any neighboring scanning.\looseness=-1

\textbf{Node cover counter $C_{Num}$.} This counter maintains how many tuples the current node will eventually cover at the leaf node. This counter is defined recursively as follows. For a Non-leaf node $N_0$, $N_{0}.C_{Num} = \sum_{N_i \in N_0.Child} N_{i}.C_{Num} $, and for a leaf node $N_1$, $N_{1}.C_{Num} = \sum_{t_j \in N_1.Keys} t_i.o(t) $. The purpose of maintaining this counter is to convert the index structure into a bidirectional mapping between the tuple's rank and the tuple's key in $log(N)$ time. Moreover, we hope that this counter can be properly isolated from the index subtree that has been changed.  This counter maintains the consistency of the bidirectional mapping before and after the update and makes the update algorithm more efficient.\looseness=-1

Apart from that, we designed the index structure similar to the B+-tree for simplicity and efficiency. Despite the significant progress in the learned index~\cite{kraska2018case,LIPP}, we still adopt this "retro" indexing implementation. The reason for using a B+-tree-style structure is that, with the help of the additional $C_{Num}$ counter and the point query mechanism of B+-tree, we can easily achieve bidirectional mapping between tuple's rank and representation in $log(N)$ time. This enables us to perform effective sampling within the rank space. In contrast, the linear functions in the learned index can only handle the injection from the tuple's key to their ranks. To support bidirectional mapping, we have to maintain an additional learned index to learn the mapping from data ranks to their keys, which is redundant and inefficient.\looseness=-1 

\subsection{Bulk-loading}
\label{Sec.bulk-loading}

Inspired by existing indexes~\cite{1979B+tree,graefe2011modern}, we propose the bulk-loading algorithm of ICE to achieve high-throughput training from massive multidimensional data. The basic idea is to sort the tuples according to their Z-values, scan layer by layer from the bottom to the top, and learn the local parameters corresponding to each layer of nodes.
The algorithm is shown in Algorithm~\ref{alg.bulkload}:\looseness=-1 

\begin{algorithm}[htb]
	\caption{ Bulk-loading.}
	\label{alg.bulkload}
	\begin{algorithmic}[1] 
		\Require
		{ Table $T$; Node fanout number: $n_f$ }
		\Ensure
		ICE structure $ICE$;
        \State {$ keys = Sort(T)$;}
		\State {$ \ell=0; s = keys.size(); levels = \phi $;}
        \While{$s > 1$} 	
        \Comment{Bottom-up training}	
        \State{$nodes = \phi; curNode = NewNode();$}
        \For{ $k\in keys$  }
        \If{$\ell=0$} \Comment{Leaf nodes}	
        \State{$curNode.C_{Num} = curNode.C_{Num} + k.o(t);$}
        \Else \Comment{Nonleaf nodes}	
        \State{$curNode.C_{Num} = curNode.C_{Num} + k.C_{Num};$}
        \EndIf
       \State{$curNode.append(k);$}
        \If{$ curNode.size() \geq n_f$}
        \State{$nodes.append(curNode);curNode = NewNode(); $}
        \EndIf
        \EndFor
        \State $keys= [ nodes[i].tupFirst, i\in \{ 0,1,\dots s-1\} ] $;
        \State{$levels[\ell] = nodes; \ell++; s=keys.size();$}

        \EndWhile
        \State{$ICE = newICE(levels);$}

		\\
		\Return {$ICE$}; 
	\end{algorithmic}
\end{algorithm}

We first sort the tuples in table $T$ under the Z-order  (line~1)  and then initialize the variables  (line~2). Subsequently, we train the ICE model bottom-up  (lines~3-14). Specifically, we maintain the local cumulative information of $C_{Num}$ layer by layer  (lines~5-9) and insert the keys from the lower layer into the nodes of the upper layer  (lines~10-12). Finally, we aggregate the layer-leveled trained information into the ICE model  (line~15) and return it.

\textbf{Complexity analysis.} \textcolor{black}{Given that each node of ICE is scanned only once during the bottom-up construction, the complexity of ICE's bottom-up process  (lines~3-14) is $O(N)$, where $N$ denotes the size of the data. Considering that the complexity of sorting the data in line~1 is in  $O(N\times log(N))$, the overall time complexity of our construction algorithm is $O(N\times log(N))$, which is determined primarily by the complexity of sorting.
Meanwhile, the bottom-up construction guarantees the depth of ICE to be $O(log(N))$.}\looseness=-1

Compared to the existing deep network's training process, the bulk-loading process of ICE is efficient. Considering that the main bottleneck of deep network training lies in the inefficiency of gradually iterating through small batches of data to compute gradients and perform backpropagation for the neural network. In contrast, the lion's share of ICE's training time lies in the $O(N \times log(N))$ complexity required for sorting multidimensional data. Specifically, if the input data is sorted, the bulk-loading algorithm mentioned above will achieve $O(N)$ performance.\looseness=-1



\subsection{Index Maintenance}
\label{Sec.CURD}

In this section, we describe instant tuple-leveled insertion, deletion, and modification on ICE within $log(N)$ time.

\begin{figure}[htbp]
	\centering
	\includegraphics[width=6.5cm]{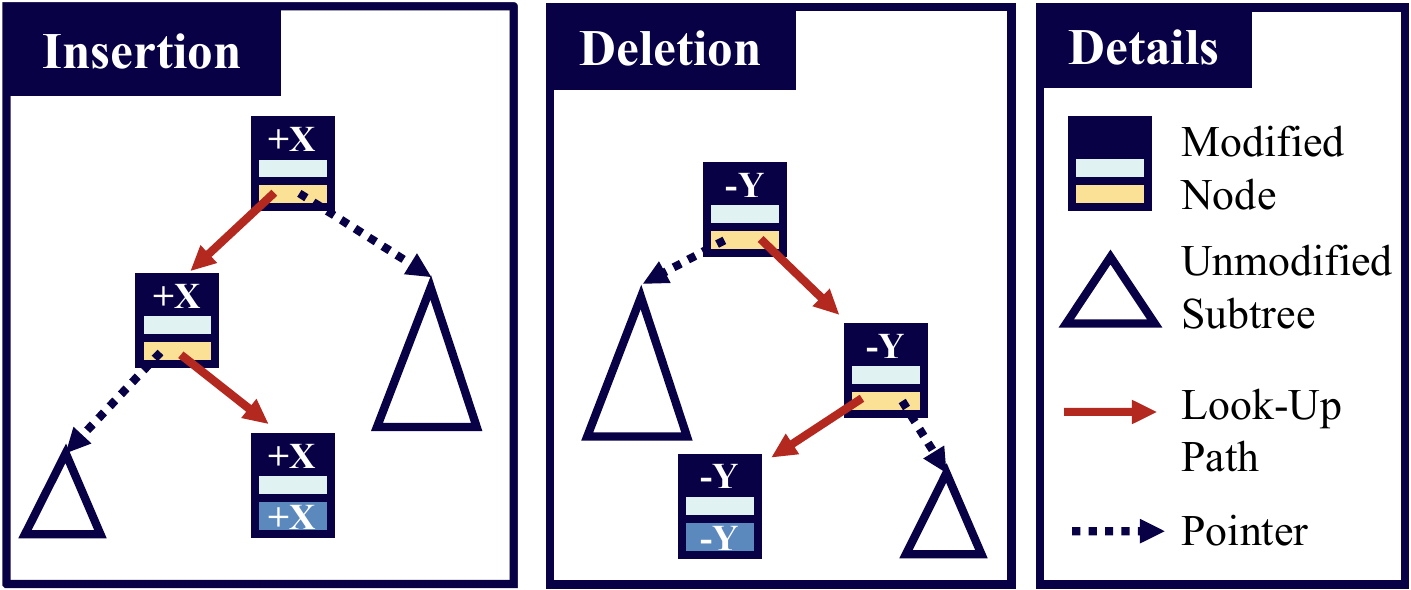}
	\caption{Insertion and deletion of ICE }
	\label{P3}
        \vspace{-1em}
\end{figure}

\textbf{Insertion and Deletion.} For insertion and deletion, as Figure~\ref{P3} shows, we first recursively look up and locate the leaf nodes that require updating and then conduct updates on the leaf nodes. If the leaf node stores the key to update, we perform arithmetic addition or subtraction on the leaf node's counter $o(t)$. If there is no available key for insertion, we allocate a new slot for the corresponding data slot address of the leaf node. If the counter indicates one during deletion, we remove such a slot from the leaf node. Subsequently, for each node on the search path, we adjust the value of its child node counter $C_{Num}$. If a node reaches a state requiring adjustment, we apply a strategy similar to that of B+-tree adjustment while maintaining the nature of the counter.\looseness=-1

\textbf{Modification.} For tuple-leveled modification, we split the operation into one deletion to the old tuple and one insertion on the new tuple. The deletion and insertion operations above ensure that node updates can be accomplished within $log(N)$ time, and the update to the counter is limited to the scale of $log(N)$. Therefore, the modification can be finished within $log(N)$ time.


\subsection{Bijection between Keys and Ranks}
\label{Sec.Tup2CDF}

In this section, we will discuss how to achieve a bidirectional mapping between the tuple representation under Z-order encoding(i.e., index key) and its corresponding rank value with the aid of ICE. The key idea of  Key2Rank mapping leverages the point query of a B+-tree, gradually accumulating the sizes of subtrees along the path of the point query to obtain the rank value of the key. Conversely, the mapping process from rank space to key space is the reverse process of Key2Rank, utilizing the accumulated subtree sizes on the scanning path to perform corresponding pruning.\looseness=-1

\begin{algorithm}[htb]
	\caption{Bijection between keys and ranks}
	\label{alg.Bijection}
	\begin{algorithmic}[1] 
		\Require
		{ICE structure $ICE$; }
		\Ensure
		{The bijection between tuple $t$'s key $k_t$ and rank $r_t$;}

        \Procedure{Key2Rank }{{$ ICE, k_t $}} \Comment{Key to rank space mapping}
        \State $r_t=0;curNode = ICE.root;$
        \While{$!curNode.isleaf$}
        \For{$N_i \in curNode.Child $}
        \If{$N_i. key < k_t$}
            \State{$r_t = r_t + N_i. C_{Num};$}
        \ElsIf{$N_i. key = k_t$}
        \State $curNode = N_i;$
        \Else
        \State{$curNode = N_i.Pre; r_t = r_t- curNode. C_{Num};$}
        \EndIf
        \EndFor
        \EndWhile
        \For{$t_i \in curNode.Keys $}
        \If{$t_i.key \leq k_t$}
            \State{$r_t = r_t + t_i.o(t);$}
        \EndIf
        \EndFor
		\Return {$r_t$}; 

        \EndProcedure
        \Procedure{Rank2Key }{{$ ICE, r_t $}} \Comment{Rank to key space mapping}
        \State $curRank=0;curNode = ICE.root;$
        \While{$!curNode.isleaf$}
        \For{$N_i \in curNode.Child $}
        \If{$curRank + N_i.C_{Num} \leq r_t$}
            \State{$curRank = curRank + N_i.C_{Num};$}
        \Else
        \State $curNode = N_i; break;$
        \EndIf
        \EndFor
        \EndWhile
        \State{$k_t = 0;$}
        \For{$t_i \in curNode.Keys $}
        \If{$curRank + t_i.o(t)< r_t$}
            \State{$curRank = curRank + t_i.o(t); k_t = t_i.key;$}
        \EndIf
        \EndFor
		\Return {$k_t$};
        \EndProcedure
	\end{algorithmic}
\end{algorithm}

In Algorithm~\ref{alg.Bijection}, we show the bidirectional mapping algorithm between index keys and key ranks. Specifically, this algorithm can be divided into two procedures: Key2Rank mapping  (lines~1-13) and Rank2Key mapping  (lines~14-25). These two procedures share similarities in their operations and can both be viewed as extensions of point queries in B+-trees. For  Key2Rank mapping, it progressively accumulates the sizes of the traversed subtrees  (line~6 and line~13) along the search path for a key $k_t$   (lines~3-13) in ICE, ultimately obtaining the rank value $r_t$ corresponding to the target tuple. In contrast, Rank2Key mapping searches for the target tuple's presence within the subtree associated with the current node based on the $C_{Num}$ information of that node and the accumulated rank sum $curRank$ traversed so far  (lines~16-21). If the target falls within the range, it switches to the corresponding subtree for further search  (line~21). Given that the depth of the ICE index is $O(log(N))$, the time complexity of both search procedures above is $O(log(N))$ in terms of search complexity.\looseness=-1

\begin{figure}[htbp]
	\centering
	\includegraphics[width=7.5cm]{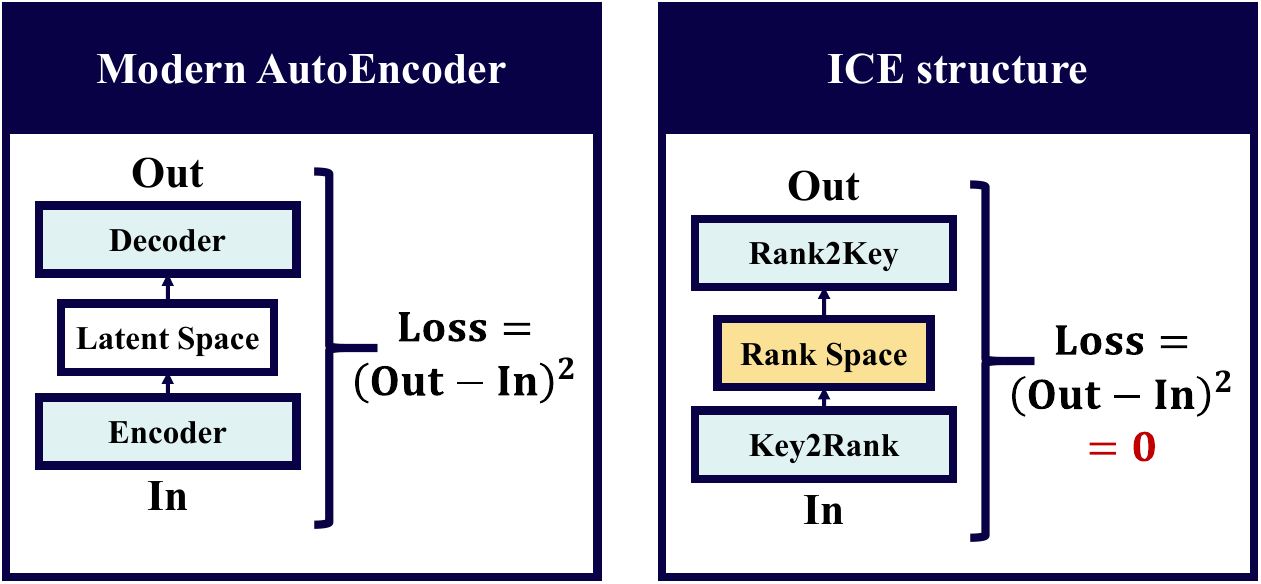}
	\caption{ICE vs. existing AutoEncoder}
	\label{P5}
        \vspace{-1em}
\end{figure}

\textbf{Discussions.} From a modern perspective, the proposed bidirectional mapping transforms the traditional index structure into an AutoEncoder with 0 loss(see Figure~\ref{P5}). In other words, given a tuple to be queried, ICE can encode it into a representation between $1$ and $N$. Meanwhile, given any positive integer from $1$ to $N$ in the rank space, 
ICE can decode and rebuild it into a tuple currently in the data table. Compared to existing AutoEncoders, ICE's latent space, namely the rank space, ensures data compactness and integrity.  That is, every tuple in the relational table corresponds one-to-one to the integer ranks in the space, and this bijection will always be consistent with the latest state of the relational table. These advantages shall thereby facilitate subsequent CE algorithms.\looseness=-1

\section{Cardinality Estimation}

In this section, we leverage the extra counters and the bijections maintained in the previous section to transform the ICE into a zero-errored AutoEncoder and sample in its latent space to accomplish CE. As shown in Figure~\ref{P4}, we first pre-filter the coarse query space by utilizing the data-skipping technique under Z-order. Subsequently, we map the filtered query space into a compact rank space using an index structure and perform efficient sampling within such rank space. Finally, we utilize the index to losslessly restore the sampled values in the rank space back into tuples, perform the final filtering in conjunction with the query, and aggregate the results. In Section~\ref{Sec.KeyFilter}, we discuss how to use the data-skipping technique to enhance the sampling efficiency. Then, in Section~\ref{Sec.CEALG}, we present our CE algorithm. Lastly, in Section~\ref{Sec.Analysis}, we conduct the  analysis on the sampling algorithm.

\begin{figure}[htbp]
	\centering
	\includegraphics[width=7.5cm]{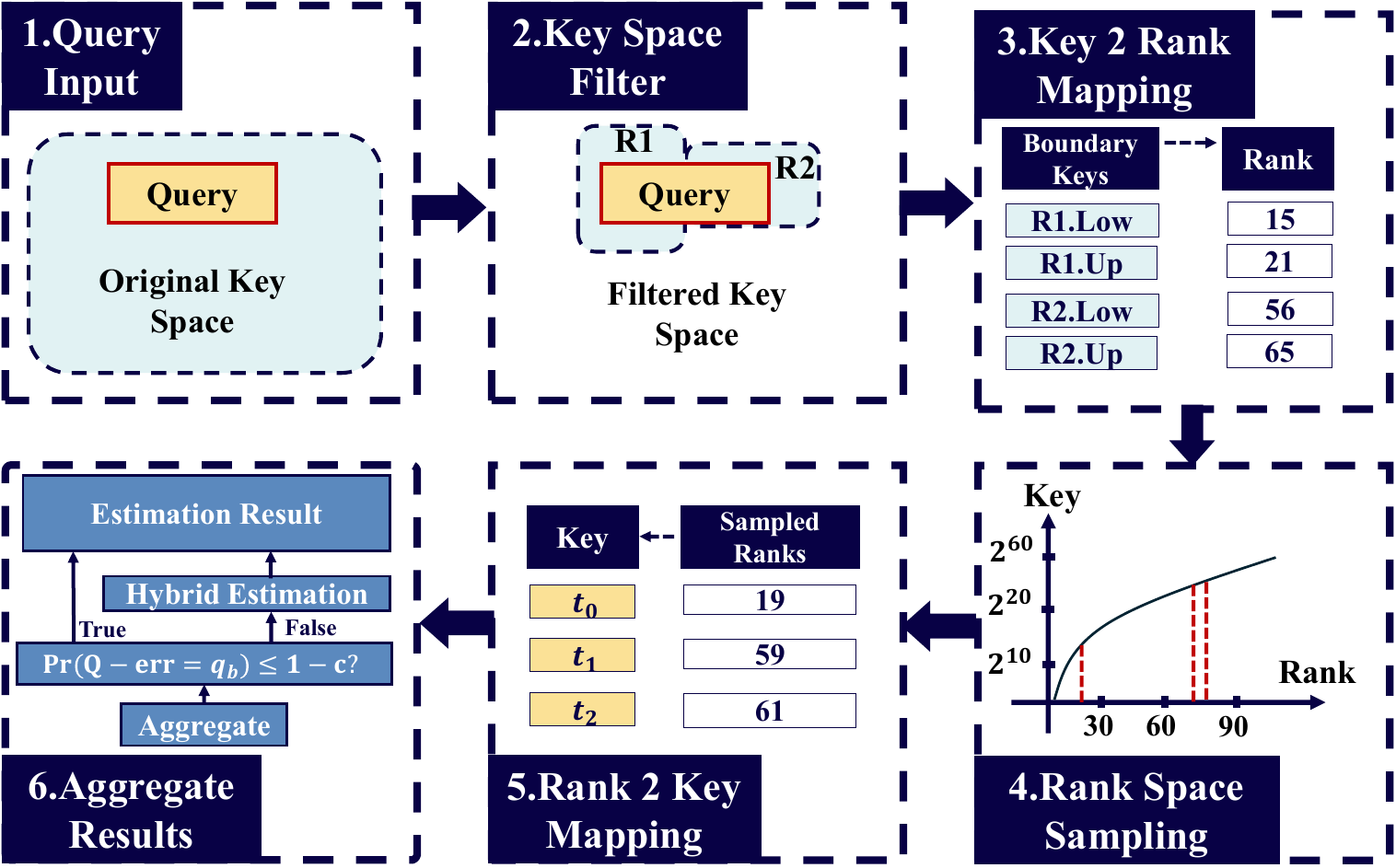}
	\caption{Overview of ICE's CE process }
        \vspace{-1.5em}
	\label{P4}
\end{figure}

\subsection{Key Space Filtering}
\label{Sec.KeyFilter}

After being mapped by a space-filling curve, a query box for range queries will encompass numerous tuples irrelevant to the query. Consequently, sampling directly within the re-mapped query box would result in low efficiency and large estimation errors. Therefore, it is necessary to pre-filter the original sampling space.\looseness=-1 

Furthermore, we observe that when handling range queries with a multidimensional index based on a space-filling curve, a common practice is to leverage the inherent properties of the curve to recursively or iteratively subdivide the query box before or during the query execution. This approach enables the "pruning" of points outside the query box, thereby reducing unnecessary scan overheads and enhancing query execution efficiency. Consequently, a straightforward intuition arises: \textbf{\textit{Can we draw inspiration from the data-skipping techniques employed in multidimensional indexing execution~\cite{ramsak2000integrating,lmsfc} to narrow the sampling space and enhance sampling efficiency?}} Thus, this section will discuss how we can improve sampling efficiency by using data-skipping techniques commonly used in multidimensional indexing.\looseness=-1

\begin{algorithm}[htb]
	\caption{Recursive filtering.}
	\label{alg.RecF}
	\begin{algorithmic}[1] 
		\Require
		{ Query box $Q$, Current recursive depth $d$.  }
		\Ensure
		A list of filtered query regions $L$;
        \Procedure{RecursiveFiltering}{$Q,d$}
        \If{$d=0$} \Comment{Initialize}
        \State{$L=[];$}
        \EndIf
        \If{$d \geq d_{max}$}  \Comment{Reached maximum depth}
        \State{$L.append(Q);$}
        \Else  \Comment{Separate and filter }
        \State{$p = findSeparationPoint(Q);$}
        \State{$Q_{1},Q_{2} = \phi;$}
        \If{$p \in Q$}
        \State{$Q_{1}.Low = Q.Low; Q_{1}.Up =p;$}
        \State{$Q_{2}.Low = p; Q_{2}.Up =Q.Up;$}
        \Else
        \State{$Q_{1}.Low = Q.Low; Q_{1}.Up = getLITMAX(p,Q);$}
        \State{$Q_{2}.Low = getBIGMIN(p,Q); Q_{2}.Up =Q.Up;$}
        \EndIf
        \State{$RecursiveFiltering(Q_1,d+1);$}
        \State{$RecursiveFiltering(Q_2,d+1);$}
        \EndIf
        \EndProcedure
        \end{algorithmic}
\end{algorithm}

We recursively implement the idea above in Algorithm~\ref{alg.RecF}. Given a query box $Q$ and the maximum search depth $d_{max}$, we recursively partition and filter the current query box. We initialize the filtered list $L$ at depth 0  (lines~2-3). If the current search depth reaches the maximum depth $d_{max}$, we add the current subdivided query to the regions list $L$  (lines~4-5). Otherwise, we progressively recursively partition and filter using the data-skipping tricks of space-filling curves  (lines~6-16). We find the appropriate partition point $p$  (line~7) and initialize the query box $Q_1,Q_2$ for further recursive searching  (line~8). If the partition point $p$ falls within the query box, we use $p$ as the left and right endpoints of $Q_1,Q_2$  (lines~9-11); otherwise, we utilize the $getBIGMIN$ and $getLITMAX$ methods from the Z-order space-filling curve\cite{ramsak2000integrating,tropf1981multidimensional} to skip the irrelevant areas and filter the current query box into new sub-region $Q_1$ and $Q_2$  (lines~12-14). Finally, we conduct a further recursive search on refined space  (lines~15-16). Given the data encoding length $n$, the above algorithm performs pre-filtering within $O(n\times2^{d_{max}})$ time.\looseness=-1


Note that the $findSeparationPoint$ function is employed in line~7 of our algorithm, which selects a separation point within the range from $Q.Low$ to $Q.Up$ to partition the subquery region. Regarding how to choose such a separation point, existing works adopt different approaches. CardIndex\cite{CardIndex} directly selects the midpoint between $Q.Low$ and $Q.Up$, \textcolor{black}{whereas  LMSFC\cite{lmsfc} formulates the selection as an optimization procedure, namely the "optimal 1-split". That is, under the condition of fixing the values of other columns, a split value is selected only on a fixed column at a time to maximize the length of the skipped gaps. Through experiments, we find out that the "optimal 1-split" is prone to falling into local optima during recursive selection, resulting in inefficient tuples filtering and relatively large estimation errors during sampling.} Consequently, we ultimately choose to select the midpoint of the query box.\looseness=-1

\subsection{Rank Space Sampling}
\label{Sec.CEALG}

In this section, we will devise a CE algorithm based on sampling in the rank space, integrating the key space filtering technique designed in the previous sections with the bijective mapping technique from the key space to the rank space learned by the index. Sampling in the rank space is primarily due to its compactness, where each rank corresponds to a unique tuple. In contrast, the key space is extremely sparse, and two adjacent tuples in the rank space may exhibit significant differences in their keys. Considering the following example:\looseness=-1

\textbf{Example.} Suppose we have a list of data: $[2^{10},
 2^{20}, 2^{30}, 2^{40}]$. If we need to sample from this sorted list to estimate the number of tuples less than $2^{20}$, sampling in the key space would require selecting two tuples less than $2^{20}$ from a space of $2^{40}$, which is highly inefficient. In contrast, the rank space of the aforementioned list is $[1,2,3,4]$. When sampling the rank space, we can directly select two tuples from a set with size equals four. This greatly enhances the sampling efficiency.\looseness=-1


\begin{algorithm}[htb]
	\caption{Index-based CE}
	\label{alg.ICESAM}
	\begin{algorithmic}[1] 
		\Require
		{ Query box $Q$; ICE structure $ICE$; Sample budget $b$; Q-error bound $q_{b}$; Confidence $c$;} 
		\Ensure
		Estimated cardinality $est$;
        \State{$L=\phi;$}
        \State{$RecursiveFiltering(Q,0);$} \Comment{ Filter key space }
        \State{$R = Key2Rank(L); $} \Comment{ Key to rank encoding }
        \State{$rSum = R.sum();$}
        \State{$samples = SampleFromRanks(R,b); $} \Comment{ Rank space sampling  }
        \State{$count=0;$}
        \For{$s_i \in samples$} \Comment{ Aggregation}
            \State{$k_i = Rank2Key(s_i);$} \Comment{ Rank to Key decoding }
            \If{$InQueryBox(k_i,Q)$}
            \State{$count= count + 1;$}
            \EndIf
        \EndFor
        \State{$est = (count/b) \times rSum;$}
        \State{$P=C_{est \cdot q_b}^{count}\times(1-\frac{b}{rSum})^{est \cdot q_b - count}\times (\frac{b}{rSum})^{count};$} 
        \If{$P>(1-c)$} \Comment{Estimates is prone to big errors}
        \State{$est = RangeQuery(ICE,Q);$} \Comment{Hybrid estimation}
        \EndIf
        \\
        \Return{$est$}
        \end{algorithmic}
\end{algorithm}

The CE algorithm is illustrated in Algorithm~\ref{alg.ICESAM}. Firstly, variable initialization is performed in line~1, followed by recursively filtering the query region within the key space  (line~2). After obtaining the filtered key space within list $L$, we encode each sub-region into the rank space and calculate the total length of these regions  (lines~3-4). Subsequently, we sample $b$ samples from these sub-regions  (line~5). During the aggregation process  (lines~7-10), we map these samples through the Rank2Key transformation (line~8), decoding them back to their tuple representations, i.e., the key values $k_i$. Then, we check whether the tuples are within the query box (line~9); if so, the counter $count$ is incremented (line~10). We estimate the results in line~11 and calculate the probability that the estimation is out of bound (line~12). When the cardinality of the query is too small, insufficient sampling cannot effectively cover the query region, and the estimated result is likely to have a high probability of significant errors (line~13). We then utilize the index execution (i.e., hybrid estimation) to conduct the last-mile search for this low-cardinality query (line~14). When getting the binomial distribution probability value in line~12, we use Gaussian approximation for simplified calculation when $est \cdot q_b$ is greater than 20. In Section~\ref{Sec.Analysis}, we prove that such last mile search could bound the Q-error within $q_b$ with a confidence of $c$.\looseness=-1

The complexity of the Algorithm~\ref{alg.ICESAM} is $O(b\times log(N))$, where $b$ is the input sampling budget, and $N$ is the total size of the data. Because we perform once point query on the index for each sample point to decode from the rank space to the data representation. \looseness=-1


\begin{figure}[htbp]
	\centering
	\includegraphics[width=6.5cm]{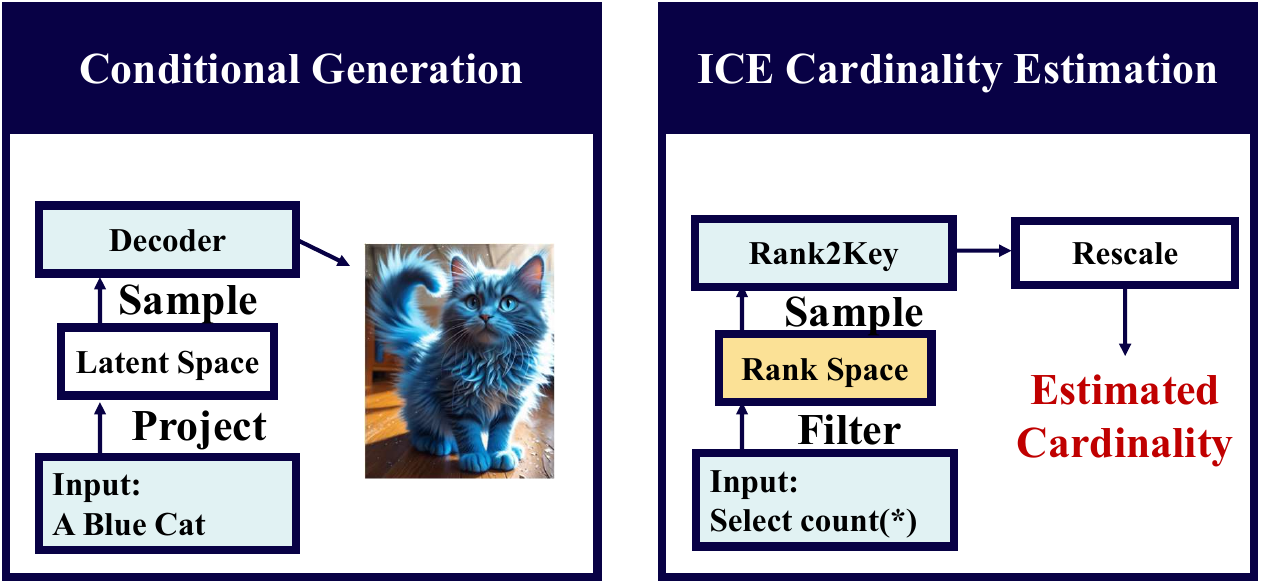}
	\caption{ICE's sampling, a generative model perspective}
	\label{P6}
\end{figure}

\textbf{Discussions.} From a generative model perspective(see Figure~\ref{P6}), the above sampling algorithm offers a more profound understanding that transcends simple sampling with replacement. Considering the discussions in Section~\ref{Sec.Tup2CDF}, ICE can be viewed as an AutoEncoder that transforms representations from the key space to its latent space, i.e., the rank space. Consequently, in Algorithm~\ref{alg.ICESAM}, the recursive filtering at line~2 essentially projects the query conditions into the ICE's latent space, while lines~6-11 employ the Rank2Key mapping to achieve conditional generation within this latent space. This conditional generation paradigm is widely adopted in computer vision fields~\cite{van2016conditional,li2019controllable,ConditionalDiffusion}, where, for instance, the semantic description of a blue cat is projected into the latent space, and a decoder subsequently samples from the projected latent space and generates an image of a blue cat based on this subspace. Similar techniques have also been applied to databases for controlled tuple generation~\cite{DifusionTup}. Compared to existing conditional generation techniques, the advantage of CE based on ICE lies in the nature of its latent space, the rank space, which is low-dimensional, ordered, and compact. The size of the conditional subspace ($rSum$) projected by the query in the rank space is straightforward to compute. We simply need to incrementally accumulate the differences in ranks between interval endpoints. This enables us to directly examine the number of generated tuples that satisfy the query conditions and rescale this count to obtain the predicted cardinality.\looseness=-1

\subsection{Analysis of ICE Sampling Algorithm}
\label{Sec.Analysis}

Next, we will prove the unbiasedness of ICE's sampling algorithm  (Algorithm~\ref{alg.ICESAM}) and give its variance analysis. Finally, we will prove the correctness of our Q-error bounding technique. \looseness=-1

\begin{theorem}
    \label{The.01}
    Given a query $Q$, the estimation result of ICE's sampling algorithm, $est$, is unbiased, i.e. $E[est] = card(Q)$.
\end{theorem}
\begin{proof}
	(Sketch):
	Regarding Algorithm~\ref{alg.ICESAM}'s second line, the recursive filtering does not exclude tuples in Q, hence preserving the estimated cardinality. ICE's bidirectional mapping between rank and key spaces ensures a one-to-one correspondence between tuples, ensuring true and predicted cardinalities align in both spaces.  Meanwhile, the hybrid estimation at line~14 of the Algorithm~\ref{alg.ICESAM} retrieves the true cardinality for low-cardinality queries, thus not affecting the unbiasedness. Therefore, we attempt to prove that the results of sampling and aggregation in the rank space in lines 5-11 of Algorithm~\ref{alg.ICESAM} are unbiased as follows.

    In the rank space, we perform sampling with replacement. We use the variable $e_i = 1$ to indicate the event that a single sample falls within the query box. The probability of this event is $Pr(e_i = 1) = {card}(Q) / rSum$, where $rSum$ is the total size of the filtered rank space (Line~4 Algorithm~\ref{alg.ICESAM}). Therefore, $E[{est}] = E[count / {b} \times rSum] = \frac{rSum}{b} \times E[\Sigma_{i=1}^{b} e_i] = \frac{rSum}{b} \times \left(  b\times \frac{card(Q)}{rSum}   \right) =  card(Q)$.
    In summary, our estimation algorithm is unbiased.

\end{proof}
In order to give the derivation of the variance, we now define the efficiency of a multidimensional index range query. The filtering efficiency $\eta$ when an index proceeds a range query $Q$ is defined as the cardinality of query $Q$ divided by the tuples scanned when executing $Q$, i.e.  $\eta = \frac{card(Q)}{\textit{\text{Tuples scanned}}}$.

From the above formula, it is not difficult to find that the least efficient way to execute a query is to scan the entire table, with an efficiency of $\eta_{scan}=card(Q)/N$, where $N$ is the full table size. An efficient multidimensional index will try to improve the efficiency as much as possible, making $\eta$ as close to 1 as possible. Therefore, similar concepts to $\eta$ are often used as a loss function for optimizing multidimensional indexes~\cite{lmsfc,ding2020tsunami,Flood}. Based on the above concepts, we derive ICE's estimation's variance:

\begin{theorem}
    \label{The.02}
    The variance of the estimation result,  $Var[est] \leq \frac{1}{b} \times card(Q)^2 \times (\frac{1}{\eta_{I}}-1)$, where $b$ is the sampling budget and $\eta_{I}$ is the efficiency when the range query $Q$ is processed on ICE. When hybrid estimation is turned off, the inequality can be an equality.
\end{theorem}

\begin{proof}	
(Sketch): We first consider the case when hybrid estimation is turned off. Considering Theorem~\ref{The.01}, our further analysis is still carried out in the rank space. Flag variable $e_i$ indicates the event that a single sample $s_i$ falls within the query box, and with a probability of $Pr(e_i=1) = card(Q)/rSum$. Since ICE can take a proactive skipping strategy that scans every tuple within the filtered rank space, the efficiency $\eta_{I}$ of ICE when executing query $Q$ is also $\frac{card(Q)}{rSum}$. We obtain that $Pr(e_i=1) = \eta_{I} $. In the process of sampling in the rank space, the aggregation result count follows a binomial distribution $B(b, \eta_{I})$. Consequently, the variance of $count$ is given by $ {Var}[{count}] = b \times \eta_{I} \times (1 - \eta_{I})$. Since the final estimated result $est$ is scaled by a factor of $rSum/b = card(Q)/\eta_{I}$ from $count$, the variance of the result $est$ when hybrid estimation turned off is:\looseness=-1

    $$  {Var}[ {est}] = \left(\frac{ {card}(Q)}{b\times\eta_{I}}\right)^2 \times  {Var}[ {count}] = \frac{1}{b} \times  {card}(Q)^2 \times \left(\frac{1}{\eta_{I}} - 1\right) $$

    Finally, when the hybrid estimation is enabled, we can execute those potentially incorrectly estimated low-cardinality queries via index search to obtain their true cardinalities without affecting the estimations of high-cardinality queries. This results in a reduction of the overall variance, which proves the inequality.

\end{proof}

\vspace{-1em}

From the above variance formula, it is evident that there are two factors independent of the query, the sampling size $b$ and the index execution efficiency $\eta_{I}$. To achieve a more precise estimation, we need to make the variance near 0. Based on Theorem~\ref{The.02},  we have two approaches of increasing the sampling budget to make the term $1/b$ approach zero ($\mathcal{A}_1$), and improving the index efficiency to make $(1/\eta_{I} - 1)$ approach zero ($\mathcal{A}_2$). We can improve the efficiency of the latter approach by increasing the recursive search depth $d_{max}$ in Algorithm~\ref{alg.RecF}. These theorems connect two widely studied problems in the field of estimation and index, as $\mathcal{A}_1$ is extensively studied in the field of CE~\cite{SampleCE2,larson2007SampleCE1}, and $\mathcal{A}_2$ is a focus of index optimization research~\cite{Flood,ding2020tsunami,lmsfc}.



Next, in Theorem~\ref{The.03}, we will prove that, without the hybrid estimation in line~14  of Algorithm~\ref{alg.ICESAM}, the probability that ICE underestimates the query's cardinality, thereby resulting in a Q-error of $q_b$, equals the probability calculated in line~12 of Algorithm~\ref{alg.ICESAM}.

\begin{theorem}
    \label{The.03}
    When hybrid estimation turned off, $Pr( Q-error = q_b ) = C_{est \cdot q_b}^{count}\times(1-\frac{b}{rSum})^{est \cdot q_b - count}\times (\frac{b}{rSum})^{count} $, where $b$, $rSum$, $count$ are the intermediate results in the Algorithm~\ref{alg.ICESAM}.
\end{theorem}

\begin{proof}(Sketch): We derive the above theorem by examining how the points in the query box are distributed within the sampling pool. Given that $est$ is underestimated by a factor of $q_b$, we know that the true cardinality is $card=q_b\times est$. Since the probability of a single sample point falling into the sampling pool during sampling is $b/rSum$, the event that $est$ has a Q-error of $q_b$ will follow a binomial distribution $B(q_b\cdot est, \frac{b}{rSum})$, therefore, $P(Q-error = q_b) = C_{q_b\cdot est}^{count}\times (1-\frac{b}{rSum})^{q_b\cdot est-count}\times (\frac{b}{rSum})^{count}$.\looseness=-1
\end{proof}

The above theorem provides us with the idea that after obtaining the sampling results in $count$, we can utilize some intermediate information(e.g., $b,rSum$) from the sampling to calculate the probability $P$ that the estimated result has a Q-error reaching the preset threshold $q_b$. The user can specify the maximum allowable Q-error $q_b$ and the probability confidence $c$. If $P$ exceeds $1-c$, it indicates that the cardinality of the query is too small, and sampling cannot obtain an accurate estimate. We will then perform an index execution (i.e., hybrid estimation) for such small-cardinality queries to obtain the true cardinality result, thereby ensuring that the model does not make estimates exceeding the maximum Q-error.\looseness=-1

\section{Experiments}
\label{sec.EXP}

In this section, we attempt to answer the following questions via our experiments.

1. Compared with the state-of-the-art cardinality estimators, under various kinds of dynamic and static environments, how does ICE perform regarding the estimation accuracy and inference time? Is the estimation robust under workload drifting? (Section~\ref{Sec.EstEva})

2.  Compared with the state-of-the-art cardinality estimators, how long does it take to train and update an ICE from massive dynamic multidimensional data? How much space does it consume? (Section~\ref{Sec.ScaEva})

3. How does the depth of the recursive filtering and the sampling budget affect the estimation accuracy of ICE? How will the selection of different split strategies affect the estimation accuracy? Does the experimental result of Q-error bounding match the theory and bound the Q-error? (Section~\ref{Sec.VarEva})

\subsection{Experimental Setup}
\label{sec.ExpSetup}

\textbf{Datasets.} We use \textcolor{black}{three} real-world datasets for experimental study on CE  tasks. \textcolor{black}{We opt to conduct experiments on these datasets as each of them has been utilized in the assessment of at least one prior study in the fields of CE~\cite{CardIndex,AreWeReady4CE,yang2019deep}.}

(1)\textit{Power}~\cite{Power}: An electric power consumption data, which owns a large domain size in all attributes (each $ \approx $ 2M). Our snapshot contains 2,049,280 tuples with 6 columns.

(2) \textit{DMV}~\cite{DMV}: A real-world dataset consisting of vehicle registration information in New York. We use the following 11 columns with widely differing data types and domain sizes (the numbers are in parentheses): record type (4), reg class(75), state (89), county (63), body type (59), fuel type (9), valid date (2101), color (225), sco ind (2), sus ind (2), rev ind(2). Our snapshot contains 12,300,116 tuples.

(3)\textit{OSM}~\cite{OSM}: Real-world geographical data. We use the dataset of Central America from OpenSteetMarket and select two columns of latitude and longitude. This dataset has 80M tuples(1.8GB in size). Our snapshot contains 80,000,000 tuples with 2 columns.  \looseness=-1

\textbf{Workloads.}  For each dataset, we adopt the conventions in the literature~\cite{ALECE} and create three different dynamic types of workloads, each of which is a random mix of insertions,
 deletions, modifications, and query statements. Each insertion/deletion/modification 
 statement only influences one tuple. The division of the query training set and test set is consistent with the literature~\cite{ALECE}. Finally, we uniformly mix queries and data update operations as follows:

 \textit{Static}:  $\#insert : \#delete : \#modify = 0:0:0$
 
 \textit{Insert-Heavy}:  $\#insert : \#delete : \#modify = 2:1:1$
 
 \textit{Update-Heavy}:  $\#insert : \#delete : \#modify = 1:1:2$
 
Regarding the construction of updating tuples, literature~\cite{ALECE} did not mention its construction process. But this is crucial as two scenarios may arise where the old model continues to perform well~\cite{AreWeReady4CE}, rendering update unnecessary: 1) All updating tuples fall outside the scope of existing queries; 2) Updates are evenly distributed across the dataset, resulting in a new data distribution that is close to the existing one.

To address these challenges, we draw inspiration from PACE~\cite{PACE} by adopting an adversarial approach. We first generate the query boxes and their cardinalities of all the original old data. Given the known query boxes and the cardinalities of the original dataset, we select a small proportion (20\%) of updating tuples based on the original data to degrade the performance of the old model. Our strategy for selecting inserting tuples is weighted sampling from the old table with replacement, assigning tuples selected by low cardinalities queries  with a higher weight. Specifically, we picked $10\%$ of the old query boxes as $W$ and formulate our update weight as $w(t) = \sum_{Q_i\in W} \theta(t\in Q_i)\times min(\frac{1}{sel(Q_i)},10^5) $.  It will maximize the probability of tuples being repeatedly selected in low-cardinality queries as much as possible. Additionally, a threshold of $10^5$ is set to prevent oversampling of the same type of tuples.  For deletion operations, we employ uniform sampling without replacement. For update operations, we decompose them into separate delete and then insert. We generated 2048 test queries for each dataset and distributed them uniformly among data modification operations. A single query is combined with the modification operation before it and is called a batch. For the query boxes of the three datasets, our cardinality distribution before and after modification is shown in Figure~\ref{Fig.QD}:\looseness=-1

\begin{figure}[t]
	\centering
	\includegraphics[width=5.5cm]{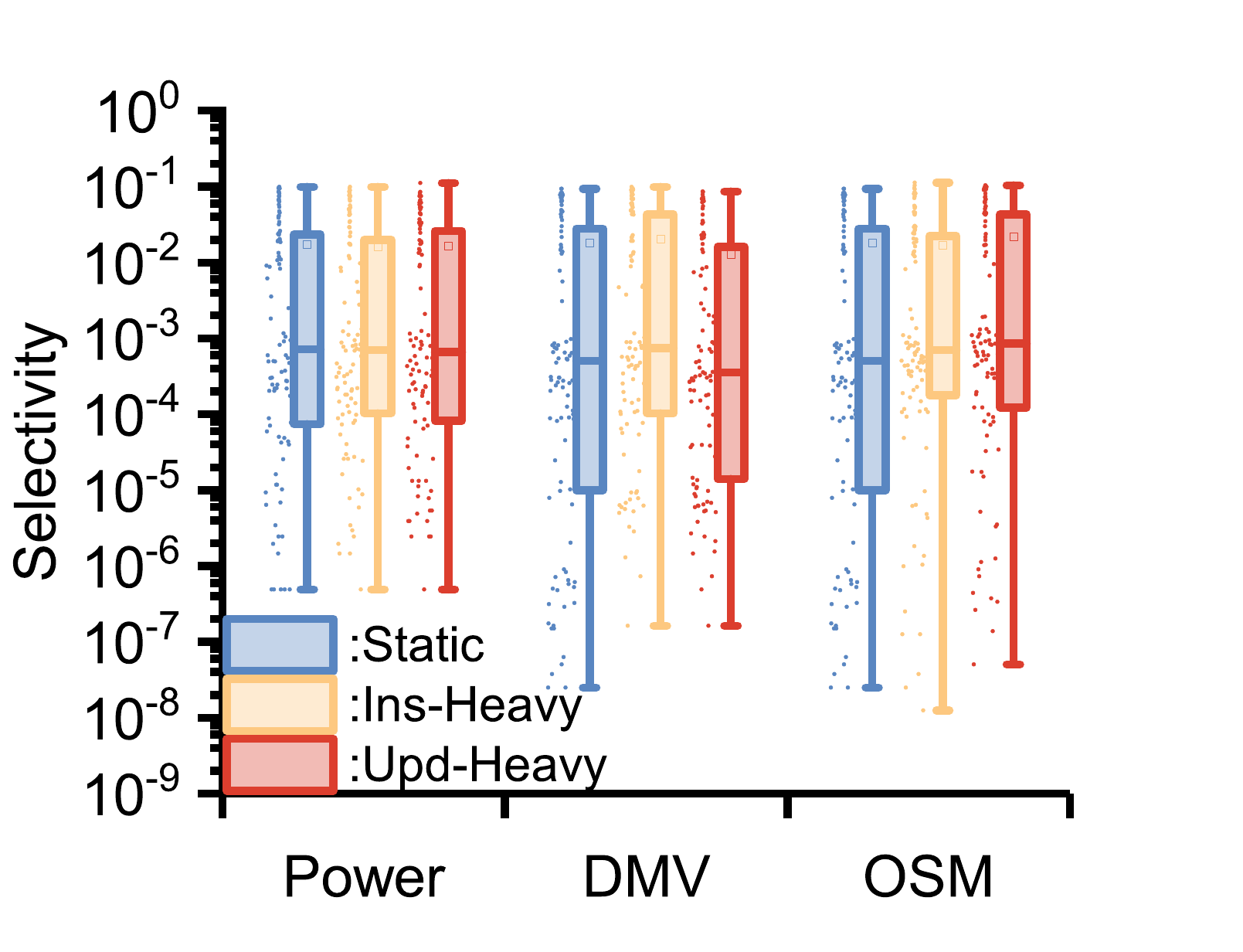}
         \vspace{-0.5em}    
 \caption{ Distribution of workload selectivity (Sampled 10\% from the query workloads)}
	\label{Fig.QD}
         \vspace{-0.5em}    
\end{figure}

\textbf{Competitors.} We choose the following baseline competitors.

(1)\textcolor{black} { \textit{Naru}~\cite{yang2019deep}: A data-driven AR network that uses progressive sampling to estimate the cardinality. }

(2) \textit{Sample}~: This approach samples several tuples in memory for CE. The sampling budget is set at 1/1000 of the original data size.\looseness=-1

(3) \textit{DeepDB}~\cite{hilprecht2019deepdb}: A data-driven method that uses SPN to estimate cardinality.

(4)\textit{ MSCN}~\cite{MSCN_Kipf2018LearnedCE}: A query-driven method that uses a multi-set convolutional network.

(5)\textit{ ALECE}~\cite{ALECE}: The state-of-the-art query-driven model that uses a transformer to learn query representation and histogram features to cardinality. It can quickly support data updates.

(6)\textit{ CardIndex}~\cite{CardIndex}: A data-driven estimator that stacks learned index with an AR model. It can use index scans to execute probed low-cardinality queries and achieve precise estimates.\looseness=-1

\textbf{Environment.} The experiments are conducted on a computer with an AMD Ryzen 7 5800H CPU, NVIDIA  RTX3060 GPU(Laptop), 64 GB RAM, and a 1 TB hard disk. We use C++ on the CPU to infer ICE and CardIndex. For the other baselines, we implement them in Python and use parallelization acceleration if possible. In other words, we leverage GPU to accelerate MSCN and Naru and enable multi-threading acceleration on CPU for CardIndex and ICE. \looseness=-1

\textbf{Evaluation metrics.} To better evaluate cardinality estimators, we adopt the following evaluation metrics: \textit{Accuracy metric}: We use Q-error to evaluate the estimator's accuracy. Q-error is defined as $ Q(E, T) = $  $ max\{\frac{E}{T},\frac{T}{E}\} $ where $ E $ is the estimated cardinality value, and  $ T $ is the real cardinality. We report each workload's entire Q-error distribution as  (50\%, 95\%, 99\%, and  Q-Max quantile). \textit{Latency metric}: We report the average time for each estimation inference and update operation. \textit{Size metric}: We report the total size of estimators.\looseness=-1

\textbf{Parameter settings.} We adopt the original settings of all the baseline methods. {We set the budget of ICE sampling at $20k$ in all datasets.} The fan-out number of the ICE index node is set to 100. And the max recursive depth $d_{max}$ is set to be $6$. We set the confidence $c$ at $1-10^{-7}$ and the maximum tolerable Q-error $q_b$ to be 20.\looseness=-1

\subsection{Estimation Evaluation}
\label{Sec.EstEva}

\begin{table*}[htbp]
\centering
\caption{  Q-errors and inference latency on dynamic and static workloads of 3 read-world datasets    }
\label{Tab.supertab}

\begin{tabular}{l|l|llll|llll|llll|l}
\multirow{2}{*}{Dataset} & \multirow{2}{*}{Method} & \multicolumn{4}{l|}{Static} & \multicolumn{4}{l|}{Insert-Heavy} & \multicolumn{4}{l|}{Update-Heavy} & \multirow{2}{*}{\begin{tabular}[c]{@{}l@{}}Inference \\Latency(ms)\end{tabular}}  \\ 
\cline{3-14}
                         &                         & 50   & 95   & 99   & Max    & 50   & 95   & 99   & Max       & 50   & 95   & 99   & Max      &                           \\ 
\hline
\multirow{8}{*}{Power }     & Naru                    & 1.11 & 2.00 & 3.7 & 16   & 1.23 & 10.3 & 501 & $1e^3$      & 1.22 & 25.8 & 741 & $1e^3$    & 9.6                      \\
                         & MSCN                   & 2.19 & 26.1 & 111 & 284   & 4.13 &58.43 & 465 & $7e^3$       &3.54 & 75.8 & 402 & $1e^3$       & \textbf{0.6}                      \\
                         & ALECE                    & 1.99 & 9.61 & 45.1 & 757   & 2.3 & 33.6 & 65 & 99.5      & 2.18  & 30.7 & 72.7 & 74   & 1.44                      \\
                         & DeepDB                  & 1.13 & 4.02 & 8.95 & 15.3   & 1.27 & 14.7 & 840 & $2e^3$      & 1.31 & 26.5 & $1e^3$ & $2e^3$      & 7.2               \\
                         & Sample                  & 1.69 & 657 & $1e^3$ & $2e^3$   & 1.52 & 722 & $2e^3$ & $4e^3$      & 1.42 & 857 & $2e^3$ & $3e^3$    & 2.38                      \\
                         & CardIndex               & 1.51 & 23.39 & 6.19 & 87   & 1.65 & 40.6 & 710 & $2e^3$      & 2.08 & 66.1 & $1e^3$ & $2e^3$     & 12                      \\

                         & ICE                  & \textbf{1.03}  & \textbf{1.46} & \textbf{2.83} & \textbf{4.33}   & \textbf{1.02} & \textbf{1.54} & \textbf{2.25} & \textbf{3.56}      & \textbf{1.02} & \textbf{1.61} & \textbf{3.04} & \textbf{8.68}     & 4.76                      \\ 
\hline
\multirow{8}{*}{DMV}   & Naru                    & 1.07 & \textbf{1.73} & 4.68 & 30   & 1.14 & 10.9  & 465 & $6e^3$     & 1.14 & 12.8 & 263 & 785     & 17.6                      \\
                         & MSCN                   & 1.8 & 30.1 & 401 & $1e^3$   & 4.55 & 112 & 399 & $4e^3$      & 3.74 & 112 & $1e^3$ & $4e^3$      & \textbf{0.71 }                     \\
                         & ALECE                    & 1.63 & 34.3 & 151 & 323   & 2.47 & 27.9 & 71 & 525    & 3.32 & 42.1 & 74.5 & 315     & 1.51                 \\
                         & DeepDB                  & 1.06 & 2.72 & 32.3 & 210   & 1.18 & 47. & 354 & $3e^3$      &  1.16 & 17.6 & 384 & $2e^4$      & 5.9                      \\
                         & Sample                  & 1.23 & 475 & $1e^3$ & $2e^3$   & 1.1 & 217 & $1e^3$ & $6e^3$      & 1.15 & 287  & $1e^3$ & $4e^3$     & 4.7                 \\
                         & CardIndex               & 1.62 & 4.68 & 7.99 & 57   & 1.71 & 57.1 & 157 & $3e^3$      & 1.43 & 74.7 & 142 & $1e^3$     & 8.46                      \\
                         & ICE                   & \textbf{1.03} & {2.35} & \textbf{{4.54}} & \textbf{17.1}   & \textbf{1.02} & \textbf{1.55} & \textbf{2.75} & \textbf{12}      & \textbf{1.03} & \textbf{2.35} & \textbf{4.54} & \textbf{11.5}     & 8.12                      \\ 
\hline
\multirow{8}{*}{OSM}     & Naru                    & 1.09 & 4.82 & 30.5 & 347   & 1.34 & 43.5 & $9e^3$ & $1e^4$      & 1.29 & 30.5 & $8e^3$ & $3e^4$     & 6.4                      \\
                         & MSCN                   & 2.96 & 314 & $2e^3$ & $8e^3$   & 5.76 & 456 & $3e^3$ & $1e^4$      & 6.47 & 657 & $3e^3$ & $6e^3$            &\textbf{0.79}          \\
                         & ALECE                    & 1.79 & 66 & 218 & 463  & 3.89 & 105 & 189 & 829      & 4.92 & 122 & 583 & $1e^3$      & 1.35                     \\
                         & DeepDB                  & 1.03 & 5.9 & 174 & $2e^3$    & 1.27 & 80.3 & $9e^3$ & $6e^4$      & 1.21 & 74.9 & $1e^4$ & $5e^4$       & 6.5                      \\
                         & Sample                  & 1.06 & 117 & 626 & $4e^3$   & 1.04 & 140 & 788 & $2e^3$      & 1.05 & 93 & 454 & $1e^3$     & 2.3               \\
                         & CardIndex               & 1.81 & 61 & 181 & $1e^3$   & 2.1 & 364 & $6e^3$ & $3e^4$      & 1.65 & 577 & $1e^4$ & $4e^4$     & 9.47                      \\
                         
                         & ICE                   & \textbf{1}  & \textbf{1.61} & \textbf{2.63} & \textbf{6.13}  &\textbf{1} & \textbf{1.43} & \textbf{3.62} & \textbf{12.5}       & \textbf{1} & \textbf{1.44} & \textbf{2} & \textbf{3.51}     & {4.31}                   
\end{tabular}
\end{table*}

\textbf{Comparison on static accuracy.} We conduct CE  tests on static workloads in Table~\ref{Tab.supertab}. We find that ICE achieves the best estimation performance across nearly all Q-error metrics. Specifically, it improves the estimation accuracy up to 50 times compared to Naru, up to 160 times compared to CardIndex, and up to 80 times compared to ALECE. This is because ICE's key space filtering trick can effectively filter out irrelevant query areas, leading to higher sampling efficiency. Meanwhile, after enabling hybrid estimation, ICE can effectively utilize the information obtained during sampling to calculate the probability that a large estimation error occurs. For small cardinality queries with high odds having large-scale errors, ICE can use fast index-scan to obtain precise results, preventing terrible mistakes and bound the maximum Q-error.\looseness=-1


\textbf{Comparison on dynamic accuracy.} We conduct tests on both Insert-Heavy and Update-Heavy workloads. In real-world environments, not all models support real-time updates; therefore, estimators with excessively long update times will directly adopt the stale model for prediction\cite{ALECE}. To determine which models will be instantly updated after the update operations, we first test the time required to update all models for one batch and report the results in Figure~\ref{Fig.BUT}. Using a threshold of 1 second per update batch in Figure~\ref{Fig.BUT}, we determine whether these baseline methods could perform instant updates.  In other words, we perform real-time updates for the four estimators: \uline{C}ard\uline{I}ndex(CI), ALECE, Sample, and ICE.  Due to the slow update time of the remaining models, we employ the old model to predict unknown query cardinalities. The results can be found in Table~\ref{Tab.supertab}.  We find out that ICE achieves the best estimation in both dynamic workloads. Regarding estimation accuracy, it is up to 2 orders of magnitude better than ALECE and up to 4 orders of magnitude better than the remaining baselines. The reasons for ICE's outstanding performance in dynamic scenarios are as follows: (1) The data distribution learned by ICE remains consistent and up-to-date with the data source, preventing ICE from predicting significantly erroneous cardinalities due to outdated models, as seen in Naru and DeepDB. (2) ICE is a data-driven model, meaning that it will not produce huge estimation errors due to out-of-distribution (OOD) phenomena on the testing workload, as ALECE and MSCN might suffer from. (3) CI's update strategy is not scalable. CI samples $O(I)$ tuples from merged data to fine-tune its AR network\cite{CardIndex}, where $I$ is the size of the update operation. However, within larger datasets like OSM and DMV, such sparse data cannot effectively update the model, which is still outdated. Considering that the AR network is also CI's root node, insufficient updates cause a chain reaction where both the point queries of the index and the progressive sampling of high cardinalities exhibit significant errors, drastically impacting its estimation accuracy and introducing errors on the order of $10^4$. (4)After opening the hybrid estimation knob, ICE can reach the lossless data distribution and use the index execution to avoid inaccurate estimation and bound the Q-error.\looseness=-1


\textbf{Comparison on estimation latency.} In terms of inference latency, the ICE's inference overhead is also low enough, as shown in Table~\ref{Tab.supertab}, ICE's inference speed, on average, is nearly twice as fast as Naru's. This is because the major part of its time consumption, sampling, does not require as much computational resources as deep AR models such as Naru. For each sampling point, only twice point queries are needed, one for Key2Rank and another for Rank2Key. The average cost of a single sampling is 2 microseconds, which can be effectively calculated even on a CPU. In contrast, deep models require expensive GPU resources, utilizing thousands of computing cores in GPUs to parallelize the inference of model parameters, resulting in extremely significant computational overhead.\looseness=-1

\begin{table}[htbp]
\caption{  Q-errors on different workload drift scenes  }
\label{tab.Robust}
\begin{tabular}{l|llll|llll}
\multirow{2}{*}{{Method}} & \multicolumn{4}{c|}{{DataDrift}} & \multicolumn{4}{c}{{QueryDrift}}  \\ \cline{2-9} 
                                     & 50     & 95       & 99       & MAX      & 50      & 95       & 99       & MAX      \\ \hline
ALECE                                &   3.65  & 54.1  & 127 & 201 & 9.27 & 137 & 634 & 668\\
MSCN                                 & 8.04 & 668 & $1e^3$ & $2e^3$ & 14.1 & 352 & $3e^3$ & $5e^3$ \\
Naru                                 & 1.11 & 2.33   & 3.98     & 5        & 1.12  & 2.06   & 2.89   & 3.33   \\
ICE                                 & \textbf{1} & \textbf{1.05}   & \textbf{1.12}     & \textbf{2.46}        & \textbf{1.02}  & \textbf{1.27}   & \textbf{1.33}   & \textbf{1.83}   \\
\end{tabular}
\end{table}

\textbf{Robustness on out-of-distribution queries.} In Table~\ref{tab.Robust}, we make two types of \uline{O}ut-\uline{O}f-\uline{D}istribution (OOD) workloads in static scenarios of the Power dataset to test the robustness of the cardinality estimators, i.e., data drift and query drift workloads.  Specifically, for the data drift workload, we sort the data by the first column, making the training queries focus on the first $50\%$ of the original data distribution, while the test queries centered on the last $50\%$.  As for the query drift workload, we concentrate the predicates of the training set mainly on the first five columns, {whereas the test set's predicates are mainly in the last column.}  We find that, on these OOD testing workloads, all data-driven methods significantly outperform the query-driven methods.  Although ALECE can utilize coarse-grained information such as histograms to mitigate these issues, it is powerless against rare predicates on unseen columns.  In contrast, data-driven cardinality estimators do not face such dilemmas, achieving over two orders of magnitude higher accuracy measured by Q-Max than query-driven models.  Furthermore, ICE achieved the best performance in both workloads.\looseness=-1

\subsection{Construction and Updating Evaluation}

\label{Sec.ScaEva}
\begin{figure}[htbp]
	\subfigure[ Update time per batch ]{
		\label{Fig.BUT}
		\includegraphics[width=0.4\columnwidth]{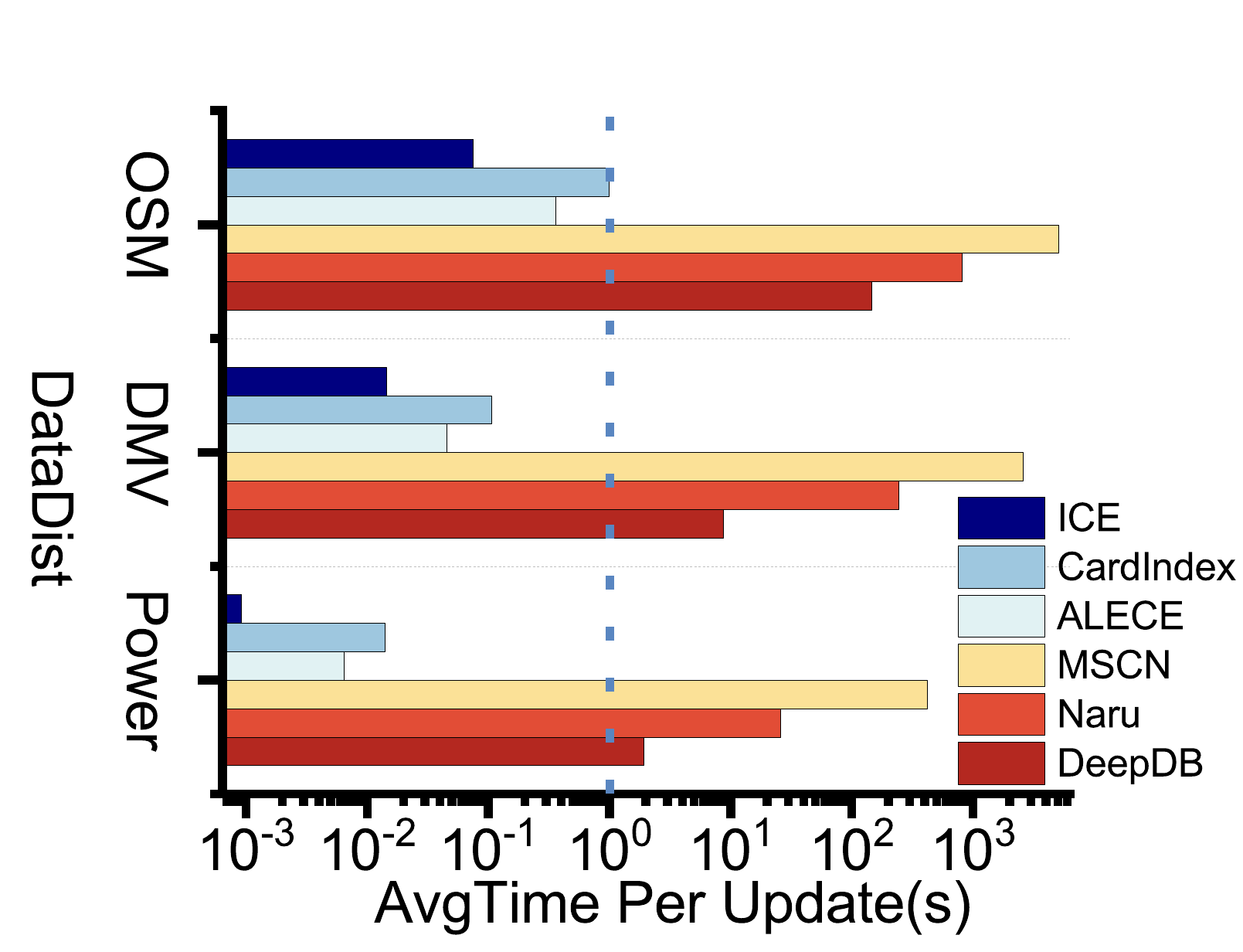}
	}
	\subfigure[ Latency breakdown ]{
		\label{Fig.ULB}
		\includegraphics[width=0.4\columnwidth]{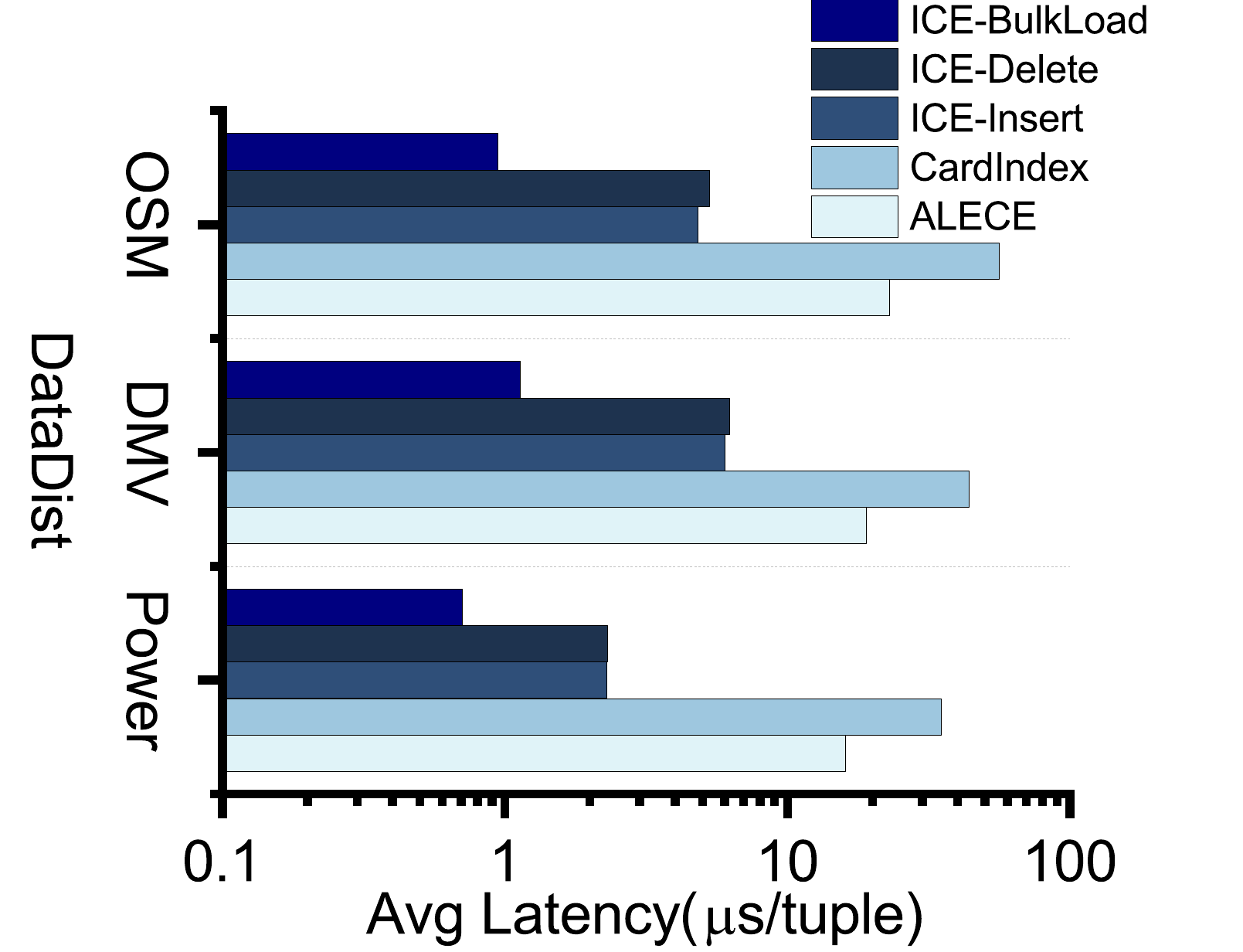  }
	}
	\caption{ Results on batched updating time and its breakdown  }
	\label{Exp03.Var}
	\vspace{-1.0em}    
\end{figure}

\textbf{Discussions on updating time.} In Figure~\ref{Fig.BUT}, we report the time required for the model to perform a single update batch. We separately select the methods that require less than 1s per update batch, decompose the batched update time, and report it separately in Figure~\ref{Fig.ULB}. We find that ICE  also achieves the fastest update speed compared to both query-driven and data-driven methods. Specifically, it is 5-6 times faster than the fastest data-driven method, CardIndex, and 2 times faster than the fastest query-driven method, ALECE. Compared to CardIndex, ICE achieves rapid updates by completely replacing the neural network with an index structure, resulting in significant advantages in parameter storage and updating. When inserting or modifying a single tuple, we only need to modify a local model parameter of $log(N)$ scale, avoiding a full update of the model parameters. Meanwhile, ALECE inefficiently maintains excessive status information when updating statistical information, leading to a non-negligible constant time overhead. At the same time, we also record the average time that bulk-loading takes to load a single tuple. We observe that "ICE-Bulkload" is about one time faster than "ICE-Insert". This is because the bulk-loading of the model does not need to adjust the tree's structure frequently, which brings lower maintenance overhead. Therefore, when rapid cold start is required, users can directly use the bulk-loading algorithm to quickly build the ICE model.\looseness=-1


\textbf{Discussions on training time and space consumption.} We report the models' training time and space consumption in Figure~\ref{Fig.TT} and Figure~\ref{Fig.MS}. We single out the time taken by the query-driven model to obtain the true cardinality label as "Label". We find out that in terms of training time, ICE achieves the fastest model training. Depending on whether the data is ordered, in the case of ordered data, ICE can be up to three orders of magnitude faster than the fastest existing model in terms of construction speed. In the case of unordered data, ICE can still be up to 40 times faster than existing models in terms of construction speed. Of course, since the distribution of the original data needs to be preserved as completely as possible, the index requires much space for maintenance, and both ICE and CardIndex occupy considerable space. In summary, the ICE model effectively trades space for accuracy and time (CE  accuracy, training time, update time, and inference latency). Considering that compared to expensive GPU resources and the limited time budget for query optimization, the memory in cloud databases is relatively cheap~\cite{sebastian2020memory,lahiri2015oraclememory}. Also, maintaining an additional index in memory for individual frequently queried tables can even accelerate the entire query processing speed. Therefore, the above trade-off of space for time and accuracy is actually acceptable in real-world scenarios.\looseness=-1

\begin{figure}[t]
	\centering
	\subfigure[ Training time]{
		\label{Fig.TT}
		\includegraphics[width=0.4\columnwidth]{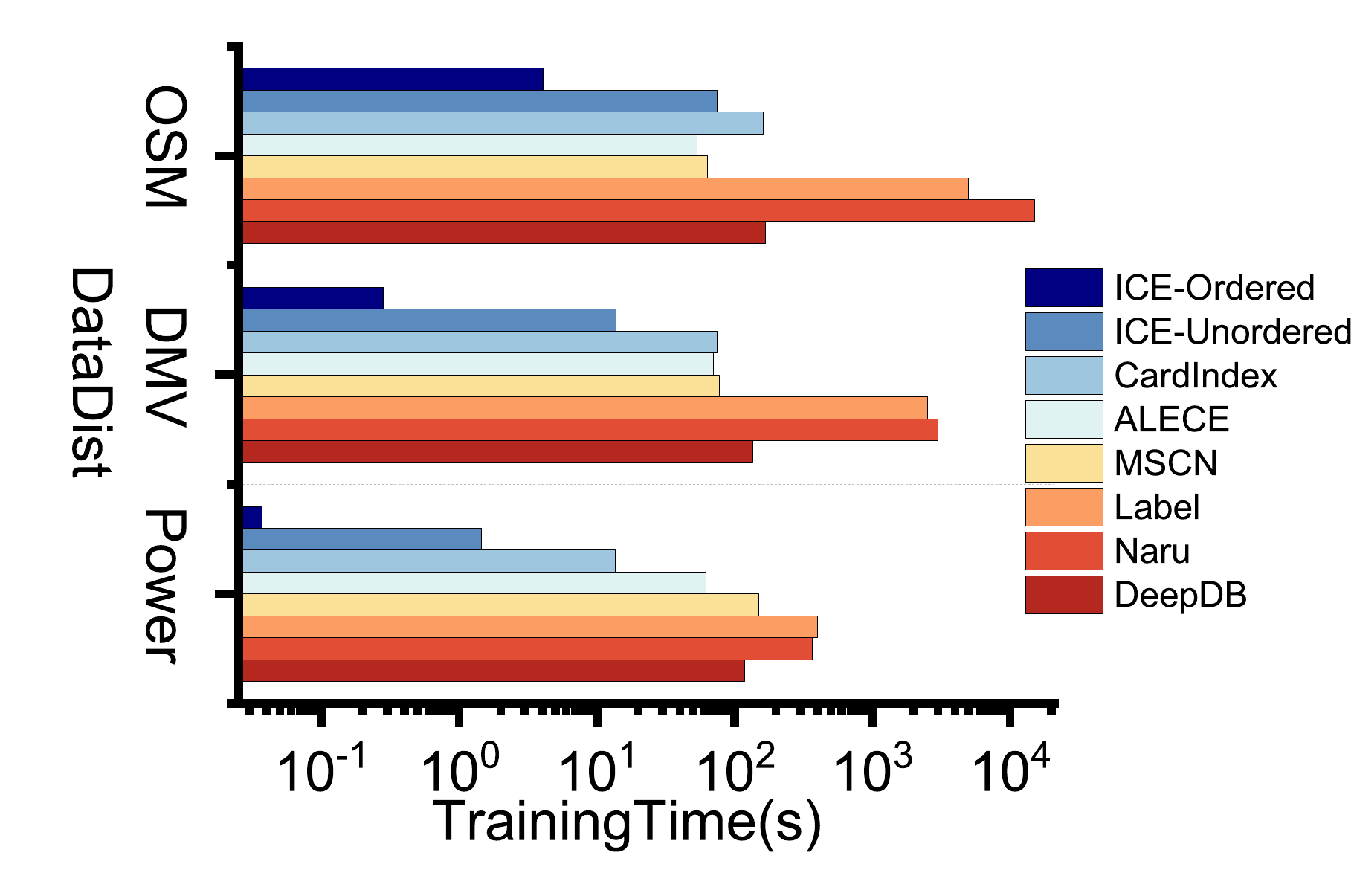  }
	}
	\subfigure[ Model size ]{
		\label{Fig.MS}
		\includegraphics[width=0.4\columnwidth]{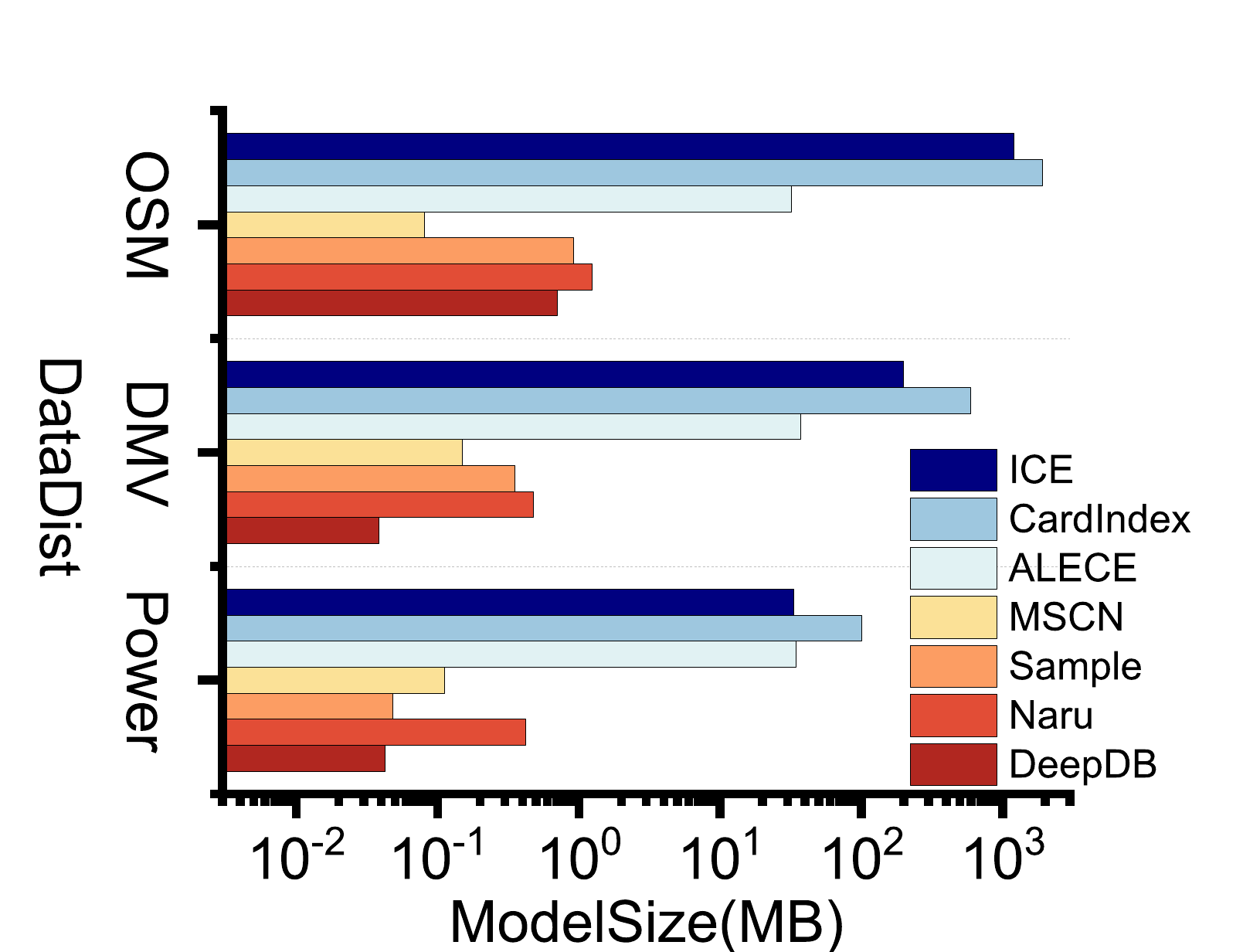}
	}
	\caption{ Results on training time and model size  }
	\label{Exp03.Var}
	\vspace{-1.0em}    
\end{figure}

\subsection{Variance Evaluation}
\label{Sec.VarEva}
In this section, we alter the parameters of our method, including search depth, sampling budget, selection strategy of the split point, Q-error bound, and confidence. We investigate their impact on the estimation accuracy and efficiency of ICE. We test on the static workload of the Power dataset. To better explore the influence of parameters on the estimation performance, we turn off the knob for hybrid estimation on the evaluation of the search depth, sampling budget, and split point selection.\looseness=-1

\begin{figure}[htbp]
	\centering
	\subfigure[  Q-errors on recursive depth]{
		\label{Fig.GRecQerror}
		\includegraphics[width=0.4\columnwidth]{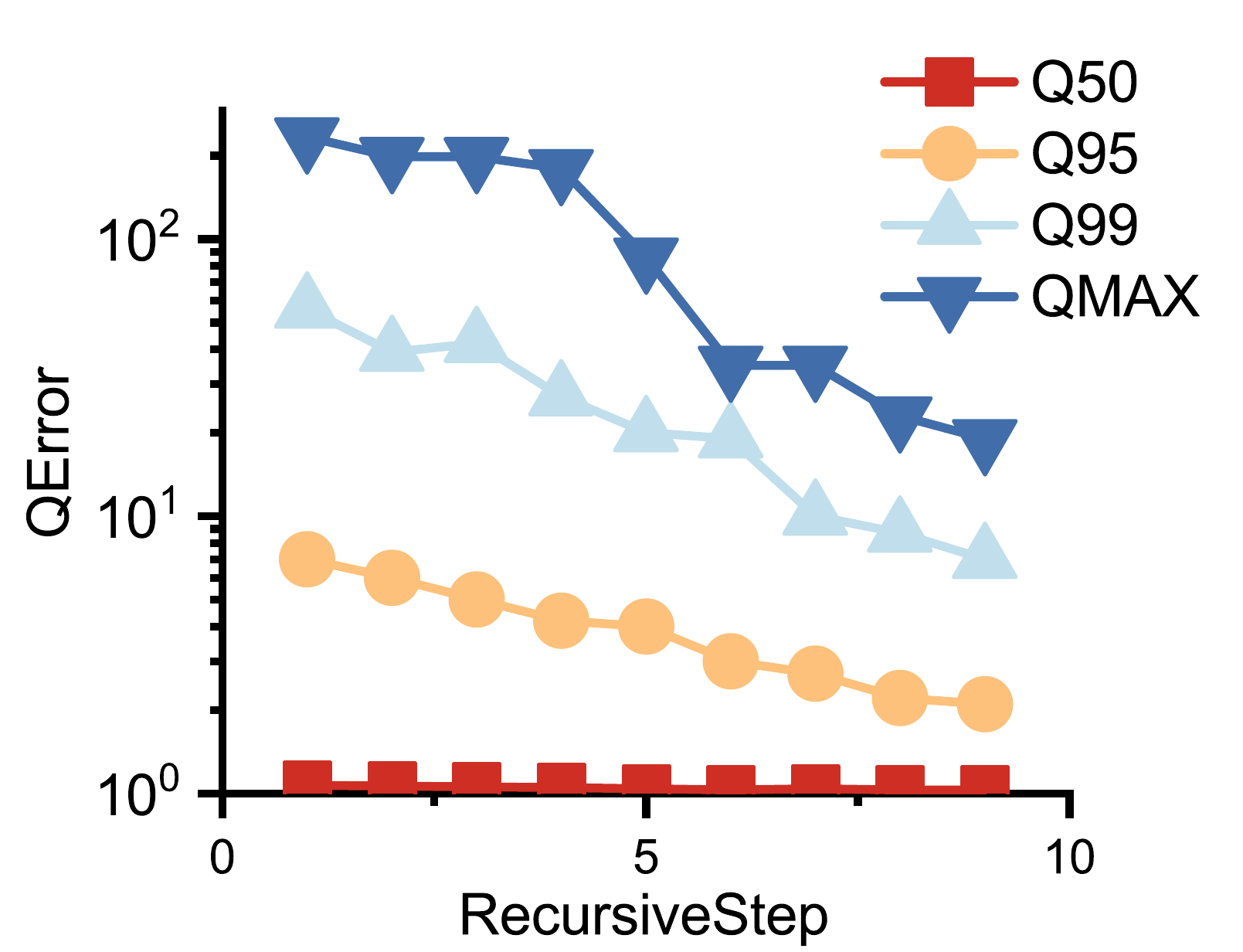  }
	}
	\subfigure[ Efficiency on recursive depth]{
		\label{Fig.GRecQMean}
		\includegraphics[width=0.4\columnwidth]{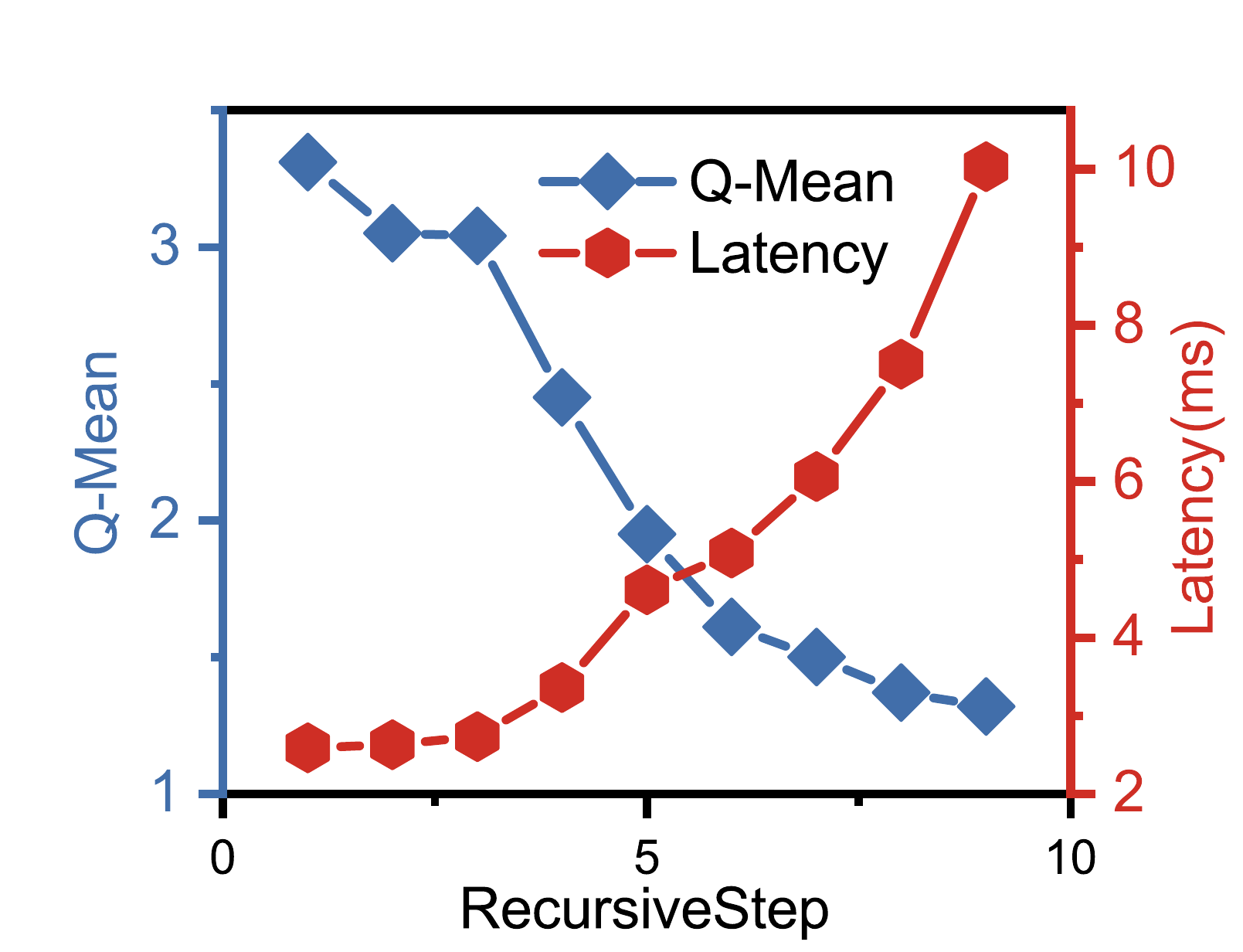}
	}
	\subfigure[  Q-errors on sample number]{
		\label{Fig.GSamQerror}
		\includegraphics[width=0.4\columnwidth]{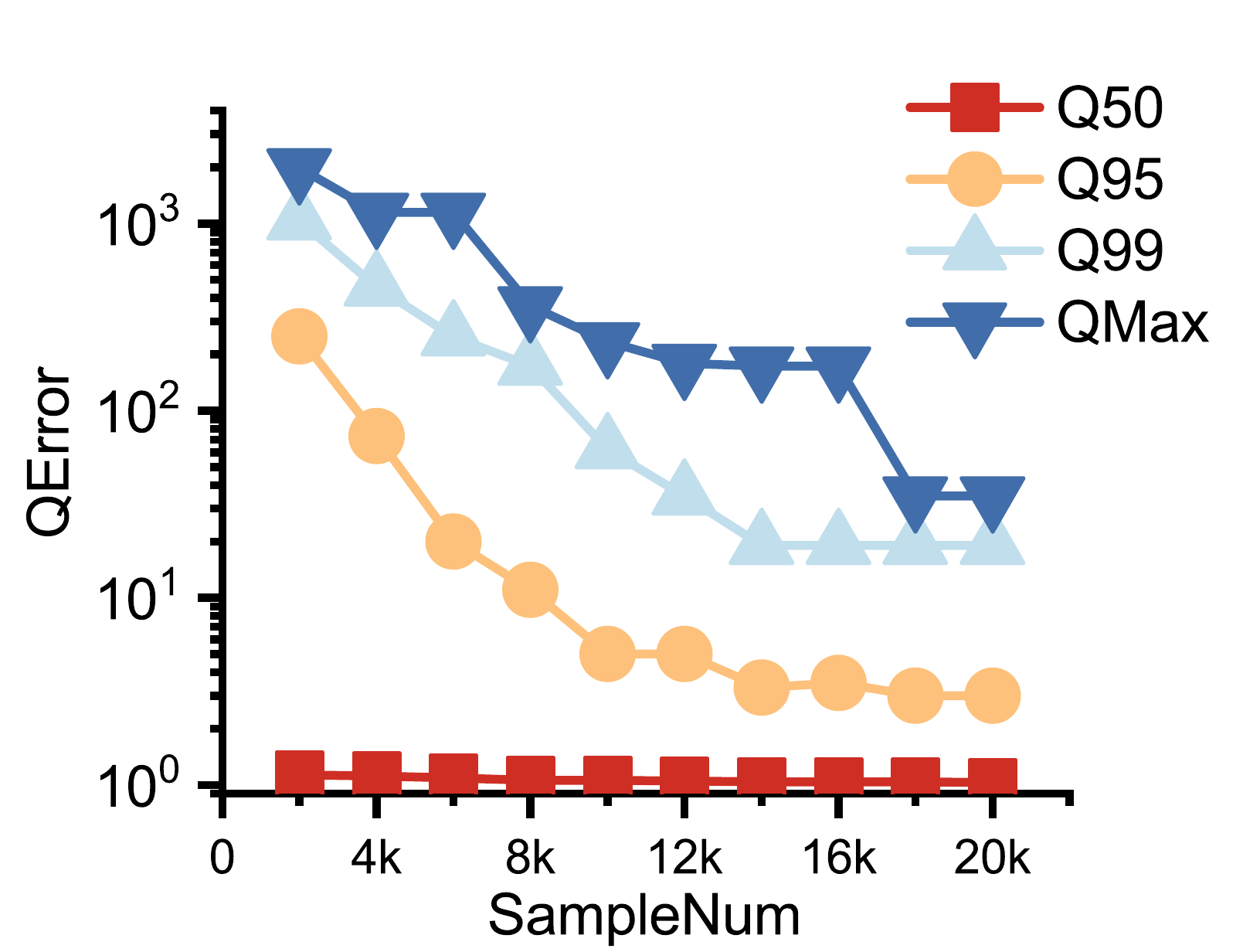  }
	}
	\subfigure[   Efficiency on sample number]{
		\label{Fig.GSamQLantency}
		\includegraphics[width=0.4\columnwidth]{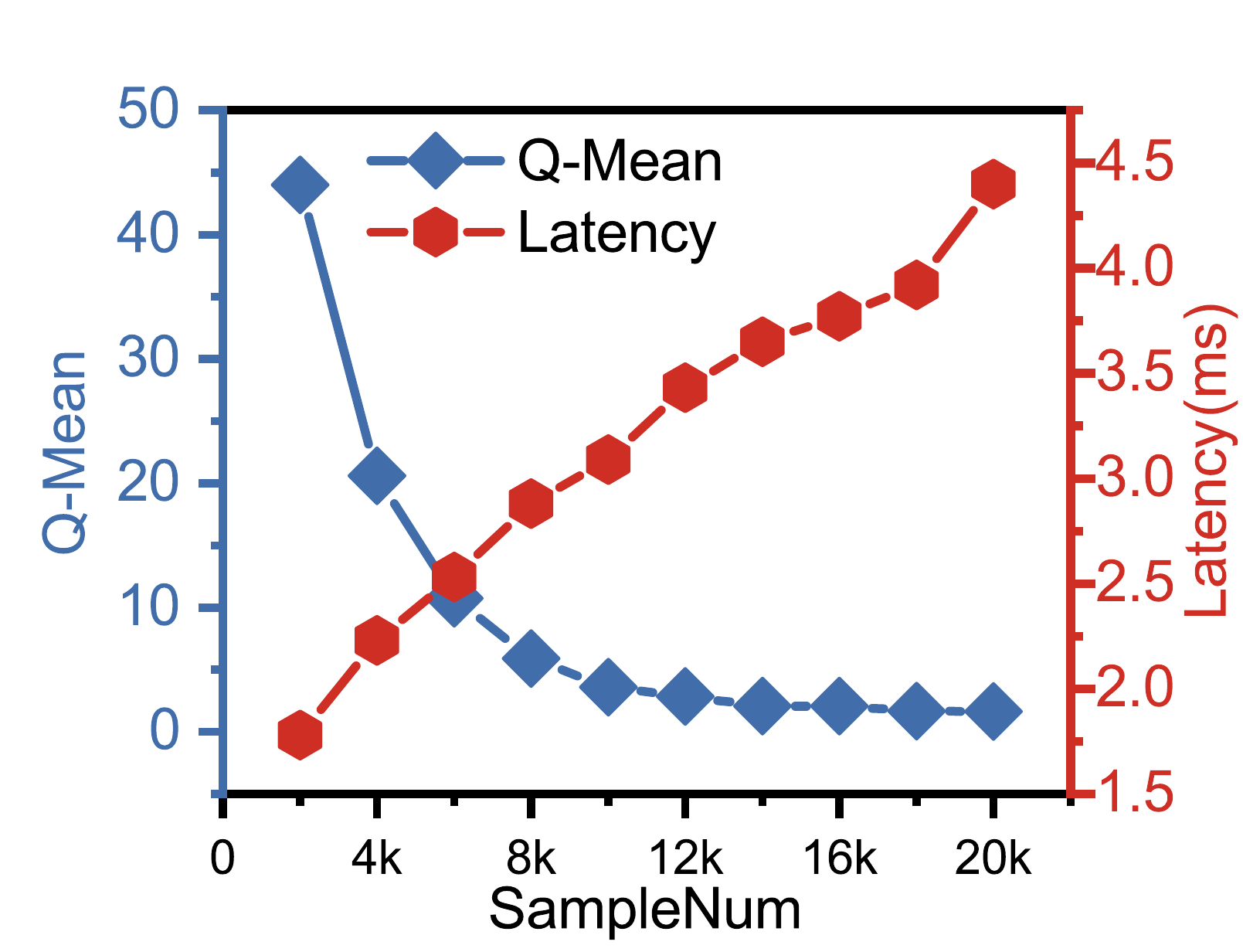}
	}
	\caption{Evaluation of sample number and recursive depth  }
	\label{Exp03.Var}
	\vspace{-1.0em}    
\end{figure}

\textbf{Effect on recursive depth and sampling number.} We explore how recursive depth and sampling budget affect the estimation quality and efficiency.   For the exploration of recursive depth, we fix the sampling budget at $20k$.  The results are in Figure~\ref{Fig.GRecQerror} and Figure~\ref{Fig.GRecQMean}. We can observe a deeper recursive partitioning depth can better filter the sampling space, leading to a lower Q-error. Specifically, when the depth increases from 1 to 9, the maximum Q-error decreases from 231 to 19, representing a reduction by an order of magnitude. {And for the evaluation of the sampling budget}, we fix the recursive depth as $6$. We report the results in Figure~\ref{Fig.GSamQerror} and Figure~\ref{Fig.GSamQLantency}. We can find out that with the increase in the number of samples, the estimation accuracy is also improved by two orders of magnitude. The reason for both experimental observations lies in the fact that, as we analyzed in Section~\ref{Sec.Analysis}, ICE's estimation variance is inversely proportional to the sampling budget $b$ and the filtering efficiency $\eta$ of the index. A deeper depth represents a stronger ability of the index to filter irrelevant areas, resulting in a higher filtering efficiency $\eta$. Higher index filtering efficiency and more sampling budgets lead to lower sampling variance and higher accuracy of estimation.\looseness=-1

\textbf{Effect on split point selection.} In Table~\ref{TABJ}, we investigate the impact of different split point selection strategies on the model's estimation quality while maintaining a sampling budget of $20k$. Our findings indicate that both sampling strategies can effectively filter the multidimensional query space. Moreover, when the recursion depth is shallow, the "Optimal 1 Split" strategy more efficiently filters a larger portion of the query space. However, as the recursion depth increases, selecting the central point strategy achieves a higher filtering efficiency. This can be attributed to the limited selection space for split points in the "Optimal 1 Split" strategy, coupled with the early termination strategy in~\cite{lmsfc}, which hinders finer-grained partitioning. Consequently, this split strategy tends to lead the filtering process into a local optimum, resulting in good efficiency only at shallow filtering depths but struggling to maintain high efficiency at deeper levels. Ultimately, this limits the improvement in estimation accuracy.\looseness=-1

\begin{table}[htbp]
\caption{Q-errors on different split strategies with different recursive depth }
\label{TABJ}
\centering
\begin{tabular}{|c|c|cccc|}
\hline
  Method&Depth       & 50th         & 95th & 99th &  MAX                                                  \\ \hline
Opt 1 Split&3    & 1.04   & 4.3 & 48& 83\\ \hline
Middle&3    & 1.06   & 5 & 42& 198\\ \hline
Opt 1 Split&6    & 1.04   & 3.7 & 26& 83\\ \hline
Middle&6    & \textbf{1.03}   & \textbf{3} & \textbf{19}& \textbf{35}\\ \hline
\end{tabular}

\end{table}

\begin{figure}[htbp]
	\vspace{-1.5em}    
	\centering
	\subfigure[  Error bound analysis]{
		\label{Fig.EBA}
		\includegraphics[width=0.4\columnwidth]{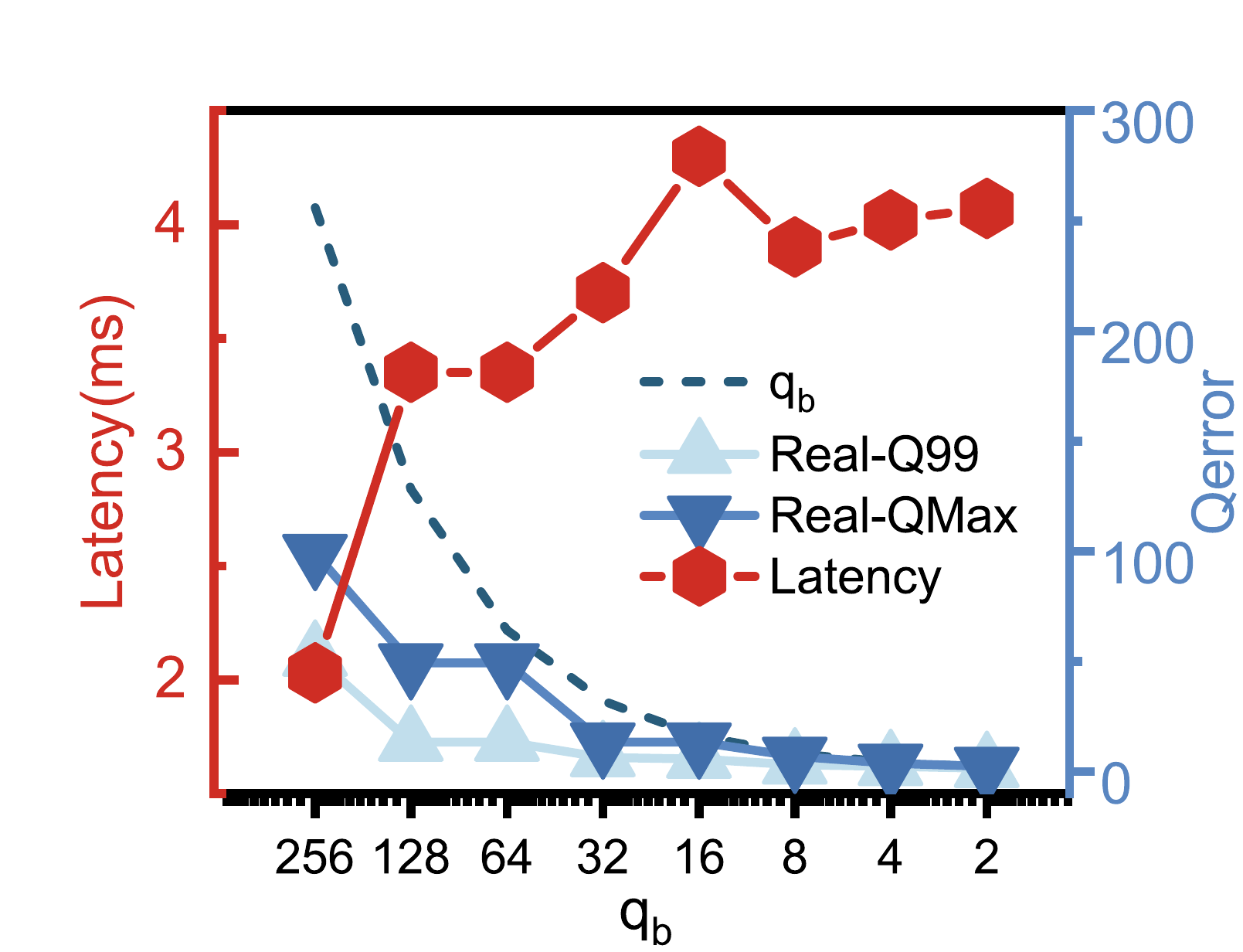  }
	}
	\subfigure[  Confidence analysis ]{
		\label{Fig.CA}
		\includegraphics[width=0.4\columnwidth]{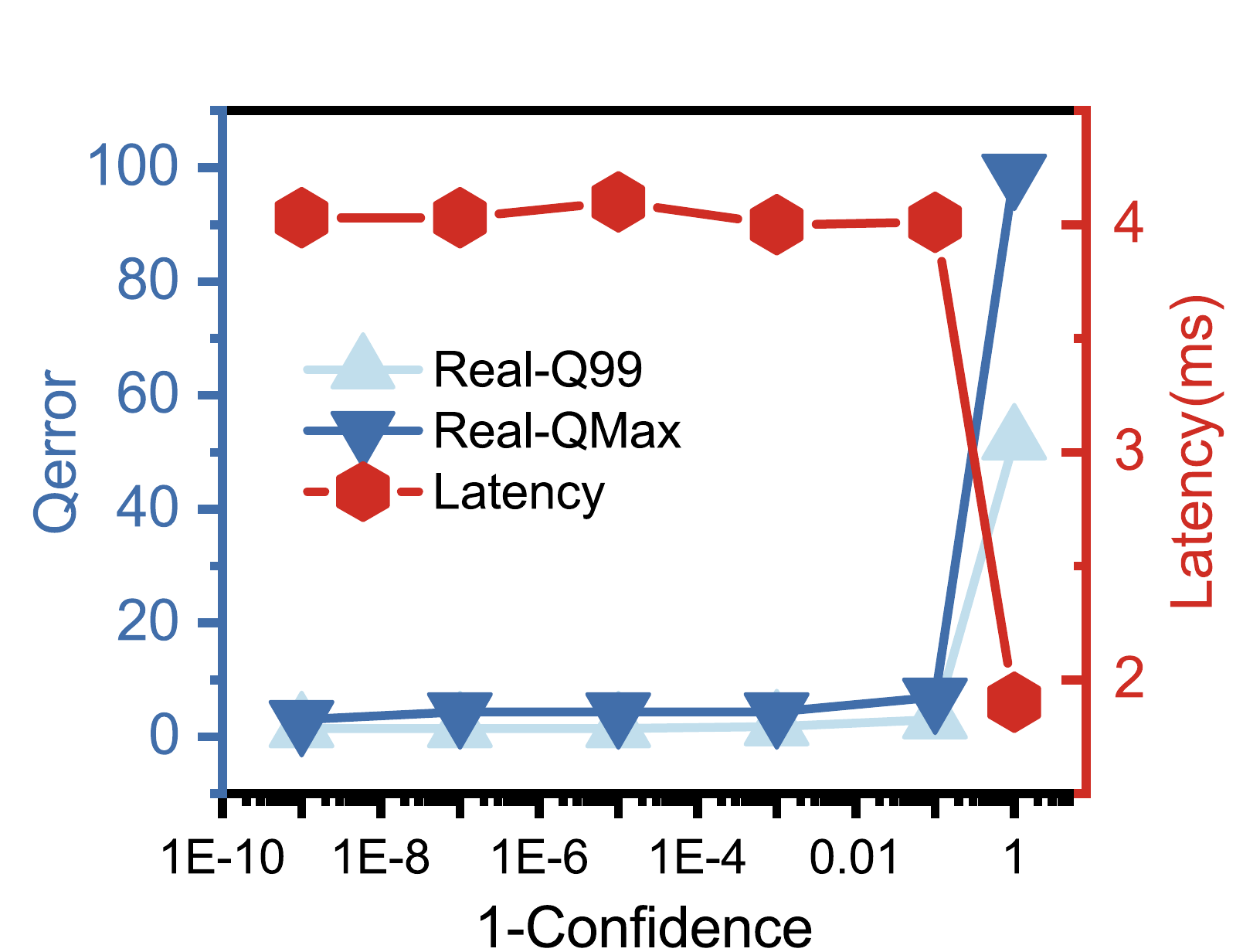}
	}
	\caption{Evaluation on the Q-error bounding}
        \label{Exp03.Var}
	\vspace{-1.5em}    
\end{figure}


\textbf{Effect on error bound and confidence.} We explore how the input Q-error bound and confidence level affect estimation quality and latency. Here, we enable hybrid estimation and conduct experiments on a small-cardinality static workload on the Power dataset, with the number of samples set to 2000. The confidence level in Figure~\ref{Fig.EBA} is $10^{-7}$, and the bound in Figure~\ref{Fig.CA} is set to 2. From Figure~\ref{Fig.EBA} and  Figure~\ref{Fig.CA}, we find that the preset bound can well bound the maximum Q-error and the 99th percentile Q-error. As $q_b$ decreases and the confidence level increases, more and more low-cardinality queries are handed over to the index for execution, resulting in more accurate estimates and relatively longer latency.\looseness=-1

\section{Conclusions and Future Work}

In this work, we designed ICE, a fully index-based cardinality estimation model, which addresses the slow update/build speed of data-driven cardinality estimators. We introduce a robust alternative approach for CE  in dynamic scenarios, where the estimation remains highly accurate without being disturbed by the drift of testing queries. Our future work will focus on two aspects: (1) Optimize index efficiency using historical workloads. According to variance analysis in Section~\ref{Sec.Analysis}, it can improve sampling efficiency and estimation quality; (2) Compress the index to reduce the model parameter scale and memory overhead.

\bibliographystyle{ACM-Reference-Format}
\bibliography{reference}

\end{document}